\documentclass[11pt]{article}
\usepackage{amsmath, amssymb, dsfont, mathrsfs}
\usepackage{amsthm}
\usepackage{fullpage}
\usepackage{float}
\usepackage{clrscode3e}
\usepackage{mdframed}
\usepackage[hidelinks]{hyperref}

\newcommand{\oldLinespread}{\linespread{1}}
\newcommand{\newLinespread}{\linespread{1.2}}
\newLinespread

\title{Nearly Optimal Deterministic Algorithm for \\ Sparse Walsh-Hadamard Transform}

\author{{\sc Mahdi Cheraghchi}\thanks{%
Email: $\langle$cheraghchi@berkeley.edu$\rangle$. Work supported in part by a Qualcomm fellowship at
Simons Institute for the Theory of Computing at UC Berkeley, and a Swiss National Science Foundation research grant PA00P2-141980. 
Part of work was done while the author was with 
MIT Computer Science and Artificial Intelligence Laboratory. } \\%
  University of California\\
  Berkeley, CA 94720
 \and %
{\sc Piotr Indyk}\thanks{
Email: $\langle$indyk@mit.edu$\rangle$.} %
 \\
  MIT\\
  Cambridge, MA 02139
}

\date{}

\date{}

\newcommand{\N}{\mathds{N}}

\newcommand{\cX}{\mathcal{X}}
\newcommand{\cY}{\mathcal{Y}}

\newcommand{\cT}{\mathcal{T}}

\newcommand{\U}{\mathcal{U}}
\newcommand{\F}{\mathds{F}}

\newcommand{\Z}{\mathbb{Z}}
\newcommand{\R}{\mathbb{R}}
\newcommand{\E}{\mathds{E}}
\newcommand{\supp}{\mathsf{supp}}

\newcommand{\eps}{\epsilon}
\newcommand{\innr}[1]{\langle #1 \rangle}

\newcommand{\xx}{{x}}
\newcommand{\hxx}{{\hat{x}}}
\newcommand{\DHT}{\mathsf{DHT}}
\newcommand{\Good}{{\mathcal{G}}}

\newtheorem{thm}{Theorem}
\newtheorem{prop}[thm]{Proposition}

\newtheorem{lem}[thm]{Lemma}

\newtheorem{coro}[thm]{Corollary}
\theoremstyle{definition}

\newtheorem{defn}[thm]{Definition}
\newtheorem{remark}[thm]{Remark}
\newtheorem{constr}{Construction}

\newenvironment{Proof}{\begin{proof}}{\end{proof}}

\newcommand{\zo}{\{0,1\}}
\newcommand{\pmo}{\{-1,+1\}}

\newcommand{\hx}{\hat{x}}
\newcommand{\tx}{\tilde{x}}

\newcommand{\Mz}{M}
\newcommand{\Gz}{G}

\floatstyle{boxed}
\newfloat{constr}{htb!}{lop}
\floatname{constr}{Construction}

\newenvironment{eqn}{\[}{\]}

\usepackage{url}
\begin{document}

\maketitle

\begin{abstract}
For every fixed constant $\alpha > 0$, we design an algorithm for computing the 
$k$-sparse Walsh-Hadamard transform of an $N$-dimensional vector $x \in \R^N$ in time $k^{1+\alpha} (\log N)^{O(1)}$.
Specifically, the algorithm is given query access to $x$ and computes a $k$-sparse $\tx \in \R^N$ satisfying
$\|\tx - \hx\|_1 \leq c \|\hx - H_k(\hx)\|_1$, for an absolute constant $c > 0$, where
$\hx$ is the transform of $x$ and $H_k(\hx)$ is its best $k$-sparse approximation. Our algorithm
is fully deterministic and only uses non-adaptive queries to $x$ (i.e., all queries are determined and performed in parallel
when the algorithm starts). 

An important technical tool that we use is a construction of nearly optimal and linear
lossless condensers which is a careful instantiation of the GUV condenser (Guruswami,
Umans, Vadhan, JACM 2009). Moreover, we design a deterministic and non-adaptive $\ell_1/\ell_1$ compressed sensing
scheme based on general lossless condensers that is equipped with a fast reconstruction algorithm
running in time $k^{1+\alpha} (\log N)^{O(1)}$ (for the GUV-based condenser) and is of independent
interest. Our scheme significantly simplifies and improves an earlier expander-based construction due to 
Berinde, Gilbert, Indyk, Karloff, Strauss (Allerton 2008).

Our methods use linear lossless condensers in a black box fashion; therefore, any future improvement
on explicit constructions of such condensers would immediately translate to improved parameters in our framework 
(potentially leading to $k (\log N)^{O(1)}$ reconstruction time with a reduced exponent in the poly-logarithmic factor,
and eliminating the extra parameter $\alpha$).

Finally, by allowing the algorithm to use randomness, while 
still using non-adaptive queries, the running time of 
the algorithm can be improved to $\tilde{O}(k \log^3 N)$.
\end{abstract}

\newpage

\tableofcontents

\newpage

\section{Introduction}

The Discrete Walsh-Hadamard transform (henceforth the Hadamard Transform or DHT) 
of a vector $x \in \R^N$, where $N = 2^n$, is a vector $\hx \in \R^N$ defined as follows:

\begin{equation} \label{eqn:had}
\hx(i) = \frac{1}{\sqrt{N}} \sum_{j \in \F_2^n} (-1)^{\innr{i, j}} x(j)
\end{equation}
where the coordinate positions are indexed by the elements of $\F_2^n$, $x(i)$ denoting
the entry at position $i \in \F_2^n$ and
the inner product $\innr{i, j}$ is over $\F_2$. 
Equivalently, the Hadamard transform is a variation of the Discrete Fourier transform (DFT) defined over the hypercube $\F_2^n$.
We use the notation $\hx = \DHT(x)$.


%
%

The standard divide and conquer approach of Fast Fourier Transform (FFT) can be applied to the Hadamard transform as well
to compute DHT in time $O(N \log N)$.  In many applications, however, most of the Fourier coefficients of a signal are small or equal to zero, i.e., the output of the DFT is (approximately) sparse. In such scenarios one can hope to design an algorithm with a running time that is {\em sub-linear} in the signal length $N$. Such algorithms would significantly improve the performance of systems that rely on processing of sparse signals.

The goal of designing efficient DFT and DHT algorithms for (approximately) sparse signals has been a subject of a large body of research, starting with  
%
%
the celebrated Goldreich-Levin theorem \cite{ref:GL89} in complexity theory\footnote{ 
This result is also known in the coding theory community as a list decoding algorithm for the Hadamard
code, and crucially used in computational learning as a part of the Kushilevitz-Mansour Algorithm for
learning low-degree Boolean functions \cite{ref:KM91}. }. The last decade has witnessed the development of several highly efficient sub-linear time sparse Fourier transform algorithms. These recent algorithms have mostly focused on the Discrete Fourier transform (DFT)
over the cyclic group $\Z_N$ (and techniques that only apply to this group), whereas some (for example, \cite{ref:SHV13}) have
focused on the Hadamard transform.  In terms of the running time, the best bounds to date were obtained in \cite{ref:HIKP12a} which showed that a $k$-sparse approximation of the DFT transform can be computed
in time $O(k (\log N)^2)$, or even in $O(k \log N)$ time if the spectrum of the signal has at most $k$ non-zero coefficients. These developments as well as some of their applications have been summarized in two surveys~\cite{ref:GST08} and~\cite{ref:GIIS14}.

While most of the aforementioned algorithms are randomized,  from both theoretical and practical viewpoints it is desirable to design {\em deterministic}  algorithms for the problem. Although such algorithms have been a subject of several works, including \cite{ref:Aka10,ref:Iwe10,ref:Iwe13}, there is a considerable efficiency gap between the deterministic sparse Fourier Transform algorithms and the randomized ones. Specifically, the best known deterministic algorithm, given in~\cite{ref:Iwe13}, finds a $k$-sparse approximation of the DFT transform of a signal in time $O(k^2 (\log N)^{O(1)})$; i.e., its running time is {\em quadratic} in the signal sparsity.  Designing a deterministic algorithm with reduced run time dependence on the signal sparsity has been recognized as a challenging open problem in the area (e.g., see Question 2 in~\cite{ref:princeton}). 

\subsection{Our result}
In this paper we make a considerable progress on this question, by designing a deterministic algorithm for DHT that runs in time $O(k^{1+\alpha} (\log N)^{O(1)})$. 
Since our main interest is optimizing the exponent of $k$ in the running time of the DHT algorithm, the reader may think of 
a parameter regime where
the sparsity parameter $k$ is not too insignificant compared to the dimension $N$ (e.g., we would like to have
$k \geq (\log N)^{\omega(1)}$, say $k \approx N^{\Theta(1)}$) so that reducing the exponent of $k$ at cost of incurring additional poly-logarithmic factors in $N$ 
would be feasible\footnote{For this reason, and in favor of the clarity and modularity of presentation, for the most part we
do not attempt to optimize the exact constant in the exponent of the $(\log N)^{O(1)}$ factor.
}.

To describe the result formally, we will consider a formulation of the problem when the algorithm is given  a query access to $\hx$ and the goal is to approximate the largest $k$ terms of $x$ using a deterministic sub-linear time algorithm\footnote{Since the Hadamard transform is its own inverse, we can interchange the roles of $x$ and $\hx$, so the same algorithm can be used to approximate the
largest $k$ terms of $\hx$ given query access to $x$. }.
More precisely, given an integer parameter $k$ and query access to $\hx$, we wish to 
compute a vector $\tx \in \F_2^N$ such that for some absolute constant $c > 0$, 
\begin{align} \label{eqn:approx}
\| \tx - x \|_1 \leq c \cdot \| H_k(x) - x \|_1,
\end{align}
where we use $H_k(x)$ to denote the approximation of $x$ to the $k$
largest magnitude coordinates; i.e., $H_k(x) \in \R^N$ is only supported on the $k$ largest
(in absolute value) coefficients of $x$ and is equal to $x$ in those positions.
Note that if the input signal $x$ has at most $k$ non-zero coefficients, then $H_k(x)=x$ and therefore the recovery is exact, i.e., $\tx = x$.
The goal formulated in \eqref{eqn:approx} is the so-called $\ell_1/\ell_1$ recovery in the sparse recovery literature.
In general, one may think of $\ell_p/\ell_q$ recovery where the norm on the left hand side (resp., right hand side) of 
\ref{eqn:approx} is $\ell_p$ (resp., $\ell_q$), such as $\ell_2/\ell_1$ or $\ell_2/\ell_2$. However, in this
work we only address the $\ell_1/\ell_1$ model as formulated in \eqref{eqn:approx} (for a survey of
different objectives and a comparison between them, see \cite{ref:FR13}).

The following statement formally captures our main result. 

\begin{thm} \label{thm:general:GUV} 
 For every fixed constant $\alpha > 0$, there is a deterministic algorithm as follows. Let $N = 2^n$ and $k \leq  N$
 be positive integers. Then, given (non-adaptive) query access to any $\hx \in \R^N$ where each coefficient
 of $\hx$ is $n^{O(1)}$ bits long, the algorithm 
 runs in time $k^{1+\alpha} n^{O(1)}$  and
 outputs $\tx \in \R^N$ that satisfies \eqref{eqn:approx} (where $\hx = \DHT(x)$)
 for some absolute constant $c > 0$.
\end{thm}

\begin{remark} \label{rem:alpha}
The parameter $\alpha$ in the above result is arbitrary as long as it is an absolute positive constant,
for example one may fix $\alpha = .1$ throughout the paper. We remark that this parameter appears not
because of our general techniques but solely
as an artifact of a particular state-of-the-art family of unbalanced expander graphs 
(due to Guruswami, Umans, and Vadhan \cite{ref:GUV09})
that we use as a part of the algorithm (as further explained below in the techniques section). 
Since we use such expander graphs as a black box, any future progress on construction of unbalanced expander graphs
would immediately improve the running time achieved by Theorem~\ref{thm:general:GUV},
potentially leading to a nearly optimal time of $k n^{O(1)}$, with linear dependence on the sparsity parameter $k$
which would be the best to hope for.   

In the running time $k^{1+\alpha} n^{O(1)}$ reported by Theorem~\ref{thm:general:GUV}, the $O(1)$
in the exponent of $n$ hides a factor depending on $1/\alpha$; i.e., the running time can be
more precisely be written as $k^{1+\alpha} n^{2/\alpha + O(1)}$. However, since $\alpha$ is taken
to be an absolute constant, this in turn asymptotically simplifies to $k^{1+\alpha} n^{O(1)}$.  Since our main focus in this
work is optimizing the exponent of $k$ (and regard the sparsity $k$ to not be too small
compared to $N$, say $k \approx N^{\Theta(1)}$), we have not attempted to optimize the exponent of $\log N$ in the running time.
However, as we will see in Section~\ref{sec:thm:main:random}, if one is willing to use randomness in the
algorithm, the running time can be significantly improved (eliminating the need for the parameter $\alpha$) 
using a currently existing family of explicit expander graphs (based on the Left-over Hash Lemma).  \qed
\end{remark}


As discussed in Remark~\ref{rem:alpha} above, our algorithm employs
state of the art constructions of explicit lossless expander graphs that to this date remain sub-optimal,
resulting in a rather large exponent in the $\log N$ factor of the asymptotic running time estimate.
Even though the main focus of this article is fully deterministic algorithms for fast recovery of the
Discrete Hadamard Transform, we further observe that the same algorithm that we develop can
be adapted to run substantially faster using randomness and sub-optimal lossless expander graphs
such as the family of expanders obtained from the Leftover Hash Lemma. As a result, we obtain
the following improvement over the deterministic version of our algorithm.

\begin{thm} \label{thm:general:random:intro} 
There is a randomized algorithm that, given integers $k, n$ (where $k \leq n$), and 
(non-adaptive) query access to any $\hx \in \R^N$ (where $N := 2^n$ and each coefficient
of $\hx$ is $O(n)$ bits long),
 outputs $\tx \in \R^N$ that, with probability at least $1-o(1)$ over the internal random
 coin tosses of the algorithm, satisfies \eqref{eqn:approx} 
 for some absolute constant $c > 0$ and $\hx = \DHT(x)$. 
Moreover, the algorithm performs a worse-case
$O(k n^3 (\log k) (\log n) ) = \tilde{O}(k (\log N)^3)$ 
arithmetic operations.
\end{thm}

\subsection{Techniques}

Most of the recent sparse Fourier transform algorithms (both randomized and deterministic) are based  on a form of ``binning''.  At a high level, sparse Fourier algorithms work by mapping (binning)  the coefficients into a small number of bins. Since the signal is sparse, each bin is likely to have only one large coefficient, which can then be located (to find its position) and estimated (to find its value). The key requirement is that the binning process needs to be performed using few samples of $\hx$, to minimize the running time. Furthermore, since the estimation step typically introduces some error, the process is repeated several times, either in parallel  (where the results of independent trials are aggregated at the end) or iteratively (where the identified coefficients are eliminated before proceeding to the next step). 

As described above, the best previous deterministic algorithm for the sparse Fourier Transform (over the cyclic group $\mathds{Z}_N$), given in~\cite{ref:Iwe13}, runs in time $k^2 \cdot (\log N)^{O(1)}$. The algorithm satisfies the guarantee\footnote{Technically, the guarantee proven in~\cite{ref:Iwe13}  is somewhat different, namely it shows that $\|\tx - x \|_2 \leq \|H_k(x)-x\|_2 +  \frac{c}{\sqrt{k}} \cdot \| H_k(x) - x \|_1$.  However, the guarantee of \eqref{eqn:approx} can be shown as well [Mark Iwen, personal communication].
In general, the guarantee of \eqref{eqn:approx} is easier to show than the guarantee in~\cite{ref:Iwe13}.  } in \eqref{eqn:approx}. The algorithm follows the aforementioned approach, where binning is implementing by {\em aliasing}; i.e., by computing a signal $y$ such that $y_j = \sum_{i\colon i \bmod p=j} x_i$, where $p$ denotes the number of bins. To ensure that the coefficients are isolated by the mapping, this process is repeated in parallel for several values of $p=p_1, p_2, \ldots p_t$. Each $p_i$ is greater than $k$ to ensure that there are more bins than elements. Furthermore, the number of different aliasing patterns $t$ must be greater than $k$ as well, as otherwise a fixed coefficient could always collide with one of the other $k$ coefficients. As a result, this approach requires more than $k^2$ bins, which results in quadratic running time. One can reduce the number of bins by resorting to randomization: The algorithm can select only some of the $p_i$'s uniformly at random and still ensure that a fixed coefficient does not collide with any other coefficient with constant probability. In the deterministic case, however, it is easy to see that  one needs to use $\Omega(k)$ mappings to isolate  each coefficient, and thus the analysis of the algorithm in~\cite{ref:Iwe13} is essentially tight.

In order to reduce the running time, we need to reduce the total number of mappings. To this end we relax the requirements imposed on the mappings. Specifically, we will require that the union of all coefficients-to-bins mappings forms a good {\em expander graph} (see section~\ref{s:prelim} for the formal definition).  Expansion is a natural property to require in this context, as it is known that there exist expanders that are induced by only $(\log N)^{O(1)}$ mappings but that nevertheless lead to near-optimal sparse recovery schemes~\cite{ref:BGIKS08}. The difficulty, however, is that  for our purpose we need to simulate those mappings on coefficients of the signal $x$, even though we can only access the spectrum $\hx$ of $x$. Thus, unlike in~\cite{ref:BGIKS08}, in our case we cannot use arbitrary ``black box'' expanders induced by arbitrary mappings.  Fortunately, there is a class of mappings that are easy to implement in our context, namely the class of  {\em linear} mappings. 

In this paper, we first show that an observation by one of the authors (as reported in~\cite{ref:mahdiPhD}) implies that there exist explicit expanders that are induced by a small number of linear mappings. From this we conclude that there exists an algorithm that makes only $k^{1+\alpha} (\log N)^{O(1)}$ queries to $\hx$ and finds a solution satisfying \eqref{eqn:approx}. However, the expander construction alone does not yield an {\em efficient} algorithm. To obtain such an algorithm, we augment the expander construction with an extra set of queries that enables us to quickly identify the large coefficients of $x$. The recovery procedure that uses those queries is iterative, and the general approach is similar to the algorithm given in Appendix~A of~\cite{ref:BGIKS08}. However, our procedure and the analysis are considerably simpler (thanks to the fact that we only use the so-called Restricted Isometry Property (RIP) for the $\ell_1$ norm instead of $\ell_p$ for $p > 1$). Moreover, our particular construction is immediately extendable for use in the Hadamard transform problem (due to the linearity properties). 

%


The rest of the article is organized as follows. Section~\ref{s:prelim} discusses notation
and the straightforward observation that the sparse DHT problem reduces to compressed sensing with
query access to the Discrete Hadamard Transform of the underlying sparse signal.
Also the notion of Restricted Isometry Property, lossless condensers, and unbalanced expander
graphs are introduced in this section. Section~\ref{sec:sample} focuses on the sample complexity;
i.e., the amount of (non-adaptive) queries that the compressed sensing algorithm (obtained by the 
above reduction) makes in order to reconstruct the underlying sparse signal. 
Section~\ref{sec:sublinear} adds to the results of the preceding section and describes our 
main (deterministic and sublinear time) algorithm to efficiently reconstruct the sparse signal
from the obtained measurements. Finally Section~\ref{sec:random} observes
that the performance of the algorithm can be improved when allowed to use randomness.
Although the main focus of this article is on deterministic algorithms, the improvement
using randomness comes as an added bonus that we believe is worthwhile to mention.

\section{Preliminaries}
\label{s:prelim}

\paragraph{Notation.}
Let $N := 2^n$ and $x \in \R^N$. We index the entries of $x$ by elements of $\F_2^n$
and refer to $x(i)$, for $i \in \F_2^n$, as the entry of $x$ at the $i$th coordinate.
The notation $\supp(x)$ is used for \emph{support} of $x$; i.e., the set of nonzero
coordinate positions of $x$. A vector $x$ is called $k$-sparse if $|\supp(x)| \leq k$.
For a set $S \subseteq [N]$ we denote by $x_S$ the $N$-dimensional vector that
agrees with $x$ on coordinates picked by $S$ and is zeros elsewhere. We thus have
$x_{\overline{S}} = x - x_S$. All logarithms in this work are to the base $2$.

\subsection*{Equivalent formulation by interchanging the roles of $x$ and $\hx$}
Recall that in the original sparse Hadamard transform problem, the algorithm is given
query access to a vector $x \in \R^N$ and the goal is to compute a $k$-sparse $\tx$ 
that approximates $\hx = \DHT(x)$. That is, 
\begin{eqn}
\| \tx - \hx \|_1 \leq c \cdot \| \hx - H_k(\hx) \|_1
\end{eqn}
for an absolute constant $c > 0$. However, since the Hadamard transform is its own inverse;
i.e., $\DHT(\hx) = x$, we can interchange the roles of $x$ and $\hx$. That is,
the original sparse Hadamard transform problem is equivalent to the problem of
having query access to the Hadamard transform of $x$ (i.e., $\hx$) and computing
a $k$-sparse approximation of $x$ satisfying \eqref{eqn:approx}. \emph{Henceforth throughout the paper,
we consider this equivalent formulation which is more convenient for establishing the
connection with sparse recovery problems.}

\paragraph{Approximation guarantees and the Restricted Isometry property:} 
We note that the equation in \eqref{eqn:approx} is similar to the $\ell_1/\ell_1$ recovery
studied in compressed sensing. In fact the sparse Hadamard transform problem as formulated above is the
same as $\ell_1/\ell_1$ when the measurements are restricted to the set of linear forms
extracting Hadamard coefficients.
Thus our goal in this work is to present a non-adaptive sub-linear time algorithm that achieves the above requirements
for all vectors $x$ and in a deterministic and efficient fashion.
It is known that the so-called Restricted Isometry Property for the $\ell_1$ norm (RIP-1) characterizes the combinatorial property needed
to achieve \eqref{eqn:approx}. Namely, we say that an $m \times N$ matrix $M$ satisfies RIP-1 of order $k$
with constant $\delta$ if for every $k$-sparse vector $x \in \R^N$,
\begin{align} \label{eqn:RIPone}
(1-\delta) \| x \|_1 \leq \| Mx \|_1 \leq (1+\delta) \| x \|_1.
\end{align}
More generally, it is possible to consider RIP-$p$ for the $\ell_p$ norm, where the norm used in the above guarantee is $\ell_p$.
As shown in \cite{ref:BGIKS08}, for any such matrix $M$, it is possible to obtain an approximation $\tx$ satisfying
\eqref{eqn:approx} from the knowledge of $Mx$. In fact, such a reconstruction can be algorithmically achieved using convex optimization
methods and in polynomial time in $N$.

\paragraph{Expanders and condensers.}
It is well known that RIP-1 matrices with zero-one entries (before normalization) are equivalent to 
adjacency matrices of unbalanced expander graphs, which are formally defined below.

\begin{defn}
A $D$-regular bipartite graph $G = (A, B, E)$ with $A$, $B$, $E$ respectively defining the set of left vertices,
right vertices, and edges, is said to be a $(k, \eps)$-unbalanced expander graph if
for every set $S \subseteq A$ such that $|S| \leq k$, we have
$|\Gamma(S)| \geq (1-\eps) D |S|$, where $\Gamma(S)$ denotes the neighborhood of $S$.  
\end{defn}

One direction of the above-mentioned characterization of binary RIP-1 matrices which is
important for the present work is the following (which we will use only for the special case $p=1$).

\begin{thm}(\cite[Theorem~1]{ref:BGIKS08}) \label{thm:exp:RIP}
Consider any $m \times N$ matrix $\Phi$ that is the adjacency matrix of a $(k, \eps)$-unbalanced
expander graph $G = (A, B, E)$, $|A| = N$, $|B| = m$, with left degree $D$, such that $1/\eps, D$ are smaller than $N$.
Then, the scaled matrix $\Phi/D^{1/p}$ satisfies the RIP-$p$ of order $k$ with constant $\delta$, 
for any $1 \leq p \leq 1 + 1/ \log N$ and $\delta = C_0 \eps$
for some absolute constant $C_0 > 1$.
\end{thm}

Unbalanced expander graphs can be obtained from the truth tables of \emph{lossless condensers},
a class of pseudorandom functions defined below.
We first recall that the \emph{min-entropy} of a distribution $\cX$ with finite support
$\Omega$ is given by $ H_\infty(\cX) := \min_{x \in \Omega}\{-\log
\cX(x)\}, $ where $\cX(x)$ is the probability that $\cX$ assigns to
the outcome $x$.   
The \emph{statistical distance} between two distributions
$\cX$ and $\cY$ defined on the same finite space $\Omega$ is given by
$ \frac{1}{2} \sum_{s \in \Omega} |\cX(s) - \cY(s)|, $ which is half
the $\ell_1$ distance of the two distributions when regarded as
vectors of probabilities over $\Omega$. Two distributions $\cX$ and
$\cY$ are said to be $\eps$-close if their statistical distance is at
most $\eps$.  

\begin{defn} \label{def:condenser}
A function $h\colon \F_2^n \times [D] \to \F_2^r$ is a $(\kappa, \eps)$-lossless condenser
if for every set $S \subseteq \F_2^n$ of size at most $2^\kappa$, the following holds: Let $X \in \F_2^n$ be 
a random variable uniformly sampled from $S$ and $Z \in [D]$ be uniformly random and independent of $X$. 
Then, the distribution of $(Z, h(X, Z))$ is $\eps$-close in statistical distance to some distribution 
with min-entropy at least $\log(D |S|)$.
A condenser is explicit if it is computable in polynomial time in $n$.
\end{defn}

Ideally, the hope is to attain $r = \kappa + \log(1/\eps) + O(1)$ and $D = O(n/\eps)$. This
is in fact achieved by a random function with high probability \cite{ref:GUV09}.
Equivalence of bipartite unbalanced expanders and lossless condensers was shown in \cite{ref:TUZ01}.
Namely, we have the following.

\begin{defn} \label{def:truthtable}
Consider a function $h\colon \F_2^n \times [D] \to \F_2^r$. 
The (bipartite) graph associated with $h$ is a bipartite graph $G=(\F_2^n, \F_2^r \times [D], E)$ with the edge set $E$ defined as follows.
For every $a \in \F_2^n$ and $(b, t) \in \F_2^r \times [D]$,
there is an edge in $E$ between $a$ and $(b, t)$ iff $h(a, t) = b$.
For any choice of $t \in [D]$, we define the function $h_t\colon \F_2^n \to \F_2^r$
by $h_t(x) := h(x, t)$. Then, the graph associated with $h_t$ is defined as the 
subgraph of $G$ induced by the restriction of the right vertices
to the set $\{(b, t)\colon b \in \F_2^r\}$.  We say that $h$ is \emph{linear in the first argument}
if $h_t$ is linear over $\F_2$ for every fixed choice of $t$.
\end{defn}

\begin{lem} (\cite{ref:TUZ01}) \label{lem:cond:expander}
A function $h\colon \F_2^n \times [D] \to \F_2^r$ is a $(\kappa, \eps)$-lossless condenser
if and only if the bipartite graph associated to $h$ is
a $(2^\kappa, \eps)$-unbalanced expander. 
\end{lem}


\section{Obtaining nearly optimal sample complexity} \label{sec:sample}

Before focusing on the algorithmic aspect of sparse Hadamard transform, we demonstrate that
deterministic sparse Hadamard transform is possible in informa\-tion-theoretic sense. That is,
as a warm-up we first focus on a sample-efficient algorithm without worrying about the running time.
The key tool that we use is the following observation whose proof is discussed in Section~\ref{sec:proof:lem:sampling}.
\begin{lem} \label{lem:sampling}
Let $h\colon \F_2^n \times [D] \to \F_2^r$, where $r \leq n$, be a function computable in time $n^{O(1)}$ and linear in the
first argument.
Let $M \in \zo^{D2^r \times 2^n}$ be the adjacency matrix of the bipartite graph associated
with $h$ (as in Definition~\ref{def:truthtable}). Then, for any $x \in \R^{2^n}$, the product $Mx$ can be computed
using only query access to $\hx = \DHT(x)$ from $D 2^r$ deterministic queries to $\hx$ and in time $D 2^r n^{O(1)}$.
\end{lem}

It is known that RIP-1 matrices suffice for sparse recovery in the $\ell_1/\ell_1$ model
of \eqref{eqn:approx}. Namely.

\begin{thm} \label{thm:RIPsparse} \cite{ref:BGIKS08}
Let $\Phi$ be a real matrix with $N$ columns satisfying RIP-1 of order $k$ with
sufficiently small constant $\delta > 0$. Then, for any vector $x \in \R^N$, there is an algorithm that
given $\Phi$ and $\Phi x$ computes an estimate $\tx \in \R^N$
satisfying \eqref{eqn:approx} in time $N^{O(1)}$.
\end{thm}

By combining this result with Lemma~\ref{lem:cond:expander},
Theorem~\ref{thm:exp:RIP}, and Lemma~\ref{lem:sampling}, we immediately 
arrive at the following result. 

\begin{thm} \label{thm:sampleOnly}
There are absolute constants $c, \eps > 0$ such that the following holds.
 Suppose there is an explicit linear $(\log k, \eps)$-lossless condenser $h\colon \F_2^n \times [D] \to \F_2^r$
 and let $N := 2^n$.
 Then, there is a deterministic algorithm running in time $N^{O(1)}$ that, given query access to $\hxx = \DHT(x) \in \R^N$,
 non-adaptively queries $\hxx$ at $D 2^r$ locations and outputs $\tx \in \R^N$ such that
 \begin{eqn}
\| \tx - \xx \|_1 \leq c \cdot \| \xx - H_k(\xx) \|_1.
 \end{eqn}
\end{thm}

\begin{Proof}
Let $M \in \zo^{D 2^r \times N}$ be the adjacency matrix of the bipartite graph associated with the condenser $h$.
By Lemma~\ref{lem:cond:expander}, $M$ represents a $(k, \eps)$-unbalanced expander graph.
Thus by Theorem~\ref{thm:exp:RIP}, $M/D$ satisfies RIP of order $k$ with constant $\delta = C_0 \eps$.
By Theorem~\ref{thm:RIPsparse}, assuming $\eps$ (and thus $\delta$) are sufficiently small constants,
it suffices to show that the product $M\xx$ for a given vector $\xx \in \R^N$
can be computed efficiently by only querying $\hxx$ non-adaptively at $D 2^r$ locations. 
This is exactly what shown by Lemma~\ref{lem:sampling}.
\end{Proof}

%

One of the best known explicit constructions of lossless condensers is due to Guruswami et al.~\cite{ref:GUV09}
that uses techniques from list-decodable algebraic codes. As observed by Cheraghchi~\cite{ref:mahdiPhD},
this construction can be modified to make the condenser linear. Namely, the above-mentioned result proves the following.

\begin{thm} \cite[Corollary~2.23]{ref:mahdiPhD}  \label{thm:GUVcondLinear}
 Let $p$ be a fixed prime power and $\alpha > 0$ be an arbitrary constant. Then, for
 parameters $n \in \N$, $\kappa \leq n \log p$, 
 and $\eps > 0$, there is an explicit linear  $(\kappa, \eps)$-lossless condenser
 $h\colon \F_{p}^n \times [D] \to \F_{p}^r$ satisfying 
  $\log D \leq (1+1/\alpha) (\log (n\kappa/\eps) + O(1)$ and 
  $r \log p \leq \log D+(1+\alpha)\kappa$. 
\end{thm}

%

For completeness, we include a proof of Theorem~\ref{thm:GUVcondLinear} in Appendix~\ref{sec:proof:GUVcondLinear}.
Combined with Theorem~\ref{thm:sampleOnly}, we conclude the following.

\begin{coro} \label{coro:sampleOnly:GUV}
 For every $\alpha > 0$ and integer parameters $N = 2^n$, $k > 0$ and parameter $\eps > 0$, there is a deterministic algorithm running 
 in time $N^{O(1)}$ that, given query access to $\hxx = \DHT(x) \in \R^N$,
 non-adaptively queries $\hxx$ at $O(k^{1+\alpha} (n \log k)^{2+2/\alpha}) = k^{1+\alpha} n^{O_\alpha(1)}$ 
 coordinate positions and outputs $\tx \in \R^N$ such that
 \begin{eqn}
\| \tx - \xx \|_1 \leq c \cdot \| \xx - H_k(\xx) \|_1,
 \end{eqn}
 for some absolute constant $c > 0$. 
\end{coro}


\subsection{Proof of Lemma~\ref{lem:sampling}} \label{sec:proof:lem:sampling}
For a vector $x \in \F_2^n$ and set $V \subseteq \F_2^n$, let $\xx(V)$ denote
the summation
\begin{eqn}
\xx(V) := \sum_{i \in V} \xx(i).
\end{eqn}
Lemma~\ref{lem:sampling} is an immediate consequence of Lemma~\ref{lem:sumVcompute} below,
before which we derive a simple proposition.

\begin{prop} \label{prop:sumV}
Let $V \subseteq \F_2^n$ be a linear space. Then for every $a \in \F_2^n$, we have
\begin{eqn}
\xx(a+V) = \frac{|V|}{\sqrt{N}} \sum_{j \in V^\perp} (-1)^{\innr{a, j}} \hxx(j).
\end{eqn}
\end{prop}

\begin{Proof}
We simply expand the summation according to the Hadamard transform formula \eqref{eqn:had} as follows.
\begin{align*}
\sum_{i \in a+V} \xx(i) &= \sum_{i \in V} \xx(i+a) \\
&= \frac{1}{\sqrt{N}} \sum_{j \in \F_2^n} \sum_{i \in V} (-1)^{\innr{a, j}} (-1)^{\innr{i, j}} \hxx(j) \\
&= \frac{|V|}{\sqrt{N}} \sum_{j \in V^\perp} (-1)^{\innr{a, j}} \hxx(j), 
\end{align*}
where the last equality uses the basic linear-algebraic fact that
\[
\sum_{i \in V} (-1)^{\innr{i, j}} = \left\{
\begin{array}{ll}
|V| & \text{if $j \in V^\perp$} \\
0 & \text{if $j \notin V^\perp$}.
\end{array}
\right.
\]
\end{Proof}

\begin{lem} \label{lem:sumVcompute}
Let $V \subseteq \F_2^n$ be a linear space and $W \subseteq \F_2^n$ be a
linear space complementing $V$. That is, $W$ is a linear sub-space such that $|W| \cdot |V| = N$ and
$V+W = \F_2^n$. Then, the vector
\begin{eqn}
v := (\xx(a+V)\colon {a \in W}) \in \R^{|W|}
\end{eqn}
can be computed in time $O(|W| \log(|W|) n)$ and by only querying $\hxx(i)$ for all $i \in V^\perp$,
assuming that the algorithm is given a basis for $W$ and $V^\perp$.
\end{lem}

\begin{Proof}
\newcommand{\vp}{V^{\perp}}
\newcommand{\spn}{\mathrm{span}}
We will use a divide and conquer approach similar to the standard Fast Hadamard Transform algorithm.
Let $r := \dim(W) = \dim(V^\perp) = n - \dim(V)$.
Fix a basis $v_1, \ldots, v_r$ of $\vp$ and a basis
$w_1, \ldots, w_r$ of $W$. For $i \in [r]$, let $\vp_i := \spn\{v_1, \ldots, v_i\}$
and $W_i := \spn\{w_1, \ldots, w_i\}$.

Let the matrix $H_i \in \pmo^{2^{i} \times 2^{i}}$ be so that the rows and columns
are indexed by the elements of $W_i$ and $\vp_i$, respectively, with the entry
at row $i$ and column $j$ defined as $(-1)^{\innr{i, j}}$. 
Using this notation, by Proposition~\ref{prop:sumV} the problem is equivalent to computing the matrix-vector
product $H_r z$ for any given $z \in \R^{2^r}$.

Note that $W_r = W_{r-1} \cup (w_r + W_{r-1})$ and similarly,
$\vp_r = \vp_{r-1} \cup (v_r + \vp_{r-1})$. Let $D_r \in \pmo^{2^{r-1}}$ be a diagonal
matrix with rows and columns indexed by the elements of $W_{r-1}$ and the 
diagonal entry at position $w \in W_{r-1}$ be defined as $(-1)^{\innr{w, v_r}}$.
Similarly, let $D'_r \in \pmo^{2^{r-1}}$ be a diagonal
matrix with rows and columns indexed by the elements of $\vp_{r-1}$ and the 
diagonal entry at position $v \in \vp_{r-1}$ be defined as $(-1)^{\innr{v, w_r}}$.
Let $z = (z_0, z_1)$ where $z_0 \in \R^{2^{r-1}}$ (resp., $z_1 \in \R^{2^{r-1}}$) is the restriction of $Z$
to the entries indexed by $\vp_{r-1}$ (resp., $v_r + \vp_{r-1}$).
Using the above notation, we can derive the recurrence
\[
H_{r} z = (H_{r-1} z_0 + D_r H_{r-1} z_1, H_{r-1} D'_r z_0 + (-1)^{\innr{v_r, w_r}} D_r H_{r-1} D'_r z_1). 
\]
Therefore, after calling the transformation defined by $H_{r-1}$ twice as a subroutine,
the product $H_{r} z$ can be computed using $O(2^r)$ operations on $n$-bit vectors.
Therefore, the recursive procedure can compute the transformation defined by $H_{r}$
using $O(r 2^r)$ operations on $n$-bit vectors.  
\end{Proof}

Using the above tools, we are now ready to finish the proof of Lemma~\ref{lem:sampling}.
Consider any $t \in [D]$.
Let $V \subseteq \F_2^n$ be
the kernel of $h_t$ and $N := 2^n$. 
Let $M^t$ be the $2^r \times N$ submatrix of $M$ consisting of rows corresponding to the
fixed choice of $t$. Our goal is to compute $M^t \cdot x$ for all fixings of $t$.
Without loss of generality, we can assume that $h_t$ is surjective. If not, 
certain rows of $M^t$ would be all zeros and the submatrix of $M^t$ obtained by removing
such rows would correspond to a surjective linear function $h'_t$ whose kernel can
be computed in time $n^{O(1)}$.

When $h_t$ is surjective,  we have $\dim V = n-r$.
Let $W \subseteq \F_2^n$ be the space of coset representatives of $V$ (i.e., $|W| = 2^r$ and $V + W = \F_2^n$). Note that
we also have $|V^\perp| = 2^r$, and that a basis for $W$ and $V^\perp$ can be computed in time $n^{O(1)}$
(in fact, $V^\perp$ is generated by the rows of the $r \times N$ transformation matrix defined by $h_t$, and a generator for
$W$ can be computed using Gaussian elimination in time $O(n^3)$). 

%

By standard linear algebra, for each $y \in \F_2^r$ there is an $a(y) \in W$ such that $h_t^{-1}(y) = a(y)+V$ 
and that $a(y)$ can be computed
in time $n^{O(1)}$. Observe that $M^t x$ contains a row for each $y$, at which the corresponding
inner product is the summation $\sum_{i \in h_t^{-1}(y)} \xx(i) = \xx(a(y)+V)$. Therefore,
the problem reduces to computing the vector $(\xx(a+V)\colon{a \in W})$ which, according to 
Lemma~\ref{lem:sumVcompute}, can be computed in time $O(r 2^r)$ in addition to the $n^{O(1)}$
time required for computing a basis for $W$ and $V^\perp$. 
By going over all choices of $t$, it follows that $Mx$ can be computed as claimed. 
This concludes the proof of Lemma~\ref{lem:sampling}.
\qed

\section{Obtaining nearly optimal reconstruction time} \label{sec:sublinear}

The modular nature of the sparse Hadamard transform algorithm presented in Section~\ref{sec:sample}
reduces the problem to the general $\ell_1/\ell_1$ sparse recovery which is of independent interest.
As a result, in order to make the algorithm run in sublinear time it suffices to design a sparse recovery algorithm
analogous to the result of Theorem~\ref{thm:RIPsparse} that runs
in sublinear time in $N$. In this section we construct such an algorithm, which is independently interesting
for sparse recovery applications. 

\subsection{Augmentation of the sensing matrix} \label{sec:augment}

A technique that has been used in the literature for fast reconstruction of exactly $k$-sparse
vectors 
 is the idea of augmenting the measurement matrix
with additional rows that guide the search process (cf.\ \cite{ref:BGIKS08}). For our application,
one obstacle that is not present in general sparse recovery is that the augmented sketch should 
be computable \emph{only with access to Hadamard transform queries}. For this reason, crucially
we cannot use any general sparse recovery algorithm as black box and have to specifically
design an augmentation that is compatible with the restrictive model of Hadamard transform queries.
We thus restrict ourselves to tensor product augmentation with ``bit selection'' matrices defined as follows,
and will later show that such augmentation can be implemented only using queries to the Hadamard
coefficients.

\begin{defn}
The bit selection matrix $B \in \zo^{n \times {N}}$ with $n$ rows and
$N = 2^n$ columns is
a matrix with columns indexed by the elements of $\F_2^n$ such that
the entry of $B$ at the $j$th row and $i$th column (where $j \in [n]$ and $i \in \F_2^n$)
is the $j$th bit of $i$.
\end{defn}

\begin{defn}
Let $A \in \zo^m \times \zo^N$ and $A' \in \zo^{m'} \times \zo^N$ be matrices.
The tensor product $A \otimes A'$ is an $mm' \times N$ binary matrix with rows
indexed by the elements of $[m] \times [m']$ such that for $i \in [m]$ and $i' \in [m']$,
the rows of $A \otimes A'$ indexed by $(i, i')$ is the coordinate-wise product
of the $i$th row of $A$ and $i'$th row of $A'$.
\end{defn}

We will use tensor products of expander-based sensing matrices with bit selection matrix, and 
extend the result of Lemma~\ref{lem:sampling} to such products.

\begin{lem} \label{lem:sampling:tensor}
Let $h\colon \F_2^n \times [D] \to \F_2^r$, where $r \leq n$, be a function computable in time $n^{O(1)}$ and linear in the
first argument, and define $N := 2^n$.
Let $M \in \zo^{D2^r \times N}$ be the adjacency matrix of the bipartite graph associated
with $h$ (as in Definition~\ref{def:truthtable}) and $M' := M \otimes B$ where $B \in \zo^{n \times N}$
is the bit selection matrix with $n$ rows. Then, for any $x \in \R^{N}$, the product $M' x$ can be computed
using only query access to $\hx$ from $O(D 2^r n)$ deterministic queries to $\hx$ and in time $D 2^r n^{O(1)}$.
\end{lem}

\begin{Proof}
For each $b \in [n]$, define $h^b\colon \F_2^n \times [D] \to \F_2^{r+1}$ to be
$h^b(x, z) := (h(x, z), x(b))$. Note that since $h$ is linear over $\F_2$, so is $h^b$ for all $b$.
Let $M''_b \in \zo^{D2^{r+1} \times N}$ be the adjacency matrix of the bipartite graph associated
with $h^b$ (as in Definition~\ref{def:truthtable}) and $M'' \in \zo^{Dn2^{r+1} \times N}$
be the matrix resulting from stacking $M''_1, \ldots, M''_n$ on top of each other. 
One can see that the set of rows of $M''$ contains the $Dn2^r$ rows of $M' = M \otimes B$.

By Lemma~\ref{lem:sampling} (applied on all choices of $h^b$ for $b \in [n]$), 
the product $M'' x$ (and hence, $M' x$) can be computed
using only query access to $\hx$ from $O(Dn 2^r)$ deterministic queries to $\hx$ and in time $O(D rn 2^r)$.
This completes the proof.
\end{Proof}

In order to improve the running time of the algorithm in Theorem~\ref{thm:sampleOnly}, we use the following result
which is our main technical tool and discussed in Section~\ref{sec:algo}.

\begin{thm} \label{thm:main}
There are absolute constants $c > 0$ and $\eps > 0$ such that the following holds.
Let $k, n, L$ ($k \leq n$ and $\log L = n^{O(1)}$) be positive integer parameters, and suppose there exists a function 
$h\colon \F_2^n \times [D] \to \F_2^r$ (where $r \leq n$) which is
an explicit $(\log(4k), \eps)$-lossless condenser. Let $M$ be the adjacency 
matrix of the bipartite graph associated with $h$ and $B$ be the bit-selection
matrix with $n$ rows and $N := 2^n$ columns. 
Then, there is an algorithm that, given $k$ and vectors $M x$ and $(M \otimes B) x$ 
for some $x \in \R^N$ (which is \emph{not} given to the algorithm and whose entries are $n^{O(1)}$ bits long), 
computes a $k$-sparse estimate $\tx$ satisfying
 \begin{eqn}
\| \tx - x \|_1 \leq c \cdot \| x - H_k(x) \|_1.
 \end{eqn}
 Moreover, the running time of the algorithm is $O(2^r D^2 n^{O(1)})$.
\end{thm}

The above result is proved using the algorithm discussed in Section~\ref{sec:algo}.
By using this result in conjunction with 
Lemma~\ref{lem:sampling:tensor} in the proof of Theorem~\ref{thm:sampleOnly}, we obtain our main result as follows.

\begin{thm} (Main) \label{thm:hadamard:sublinear}
There are absolute constants $c > 0$ and $\eps > 0$ such that the following holds.
Let $k, n$ ($k \leq n$) be positive integer parameters, and suppose there exists a function 
$h\colon \F_2^n \times [D] \to \F_2^r$ (where $r \leq n$) which is
an explicit $(\log(4k), \eps)$-lossless condenser and is linear in the first argument. 
 Then, there is a deterministic algorithm running in time $2^r D^2 n^{O(1)}$ that, given (non-adaptive) query access to $\hxx \in \R^N$
 (where $N := 2^n$, and each entry of $\hxx$ is $n^{O(1)}$ bits long),
 outputs $\tx \in \R^N$ such that
 \begin{eqn}
\| \tx - \xx \|_1 \leq c \cdot \| \xx - H_k(\xx) \|_1.
 \end{eqn}
\end{thm}

\begin{Proof}
We closely follow the proof of Theorem~\ref{thm:sampleOnly}, but in the proof use
Theorem~\ref{thm:main} instead of Theorem~\ref{thm:RIPsparse}.

Since each entry of $\hxx$
is $n^{O(1)}$ bits long and the Hadamard transform matrix (after normalization)
only contains $\pm 1$ entries, we see that each entry of $\sqrt{N} \xx$ is $n^{O(1)}$ bits long as well.

Let $M$ be the adjacency matrix of the bipartite expander graph associated with $h$,
$B$ be the bit selection matrix with $n$ rows, and $M' := M \otimes B$. 
By the argument of Theorem~\ref{thm:sampleOnly}, the product $M\xx$ can be
computed in time $2^r D n^{O(1)}$ only by non-adaptive query access to $\hxx$.
Same is true for the product $M' \xx$ using a similar argument and using 
Lemma~\ref{lem:sampling:tensor}.
Once computed, this information can be passed to the algorithm guaranteed by Theorem~\ref{thm:main}
to compute the desired estimate on $\xx$.
\end{Proof}

Finally, by using the condenser of Theorem~\ref{thm:GUVcondLinear} in the above theorem, we immediately obtain 
Theorem~\ref{thm:general:GUV} as a corollary, which is restated below.

\let\oldthm\thethm
\renewcommand{\thethm}{\ref{thm:general:GUV}}
\begin{thm}[restated] 
 For every fixed constant $\alpha > 0$, there is a deterministic algorithm as follows. Let $N = 2^n$ and $k \leq  N$
 be positive integers. Then, given (non-adaptive) query access to any $\hx \in \R^N$ where each coefficient
 of $\hx$ is $n^{O(1)}$ bits long, the algorithm 
 runs in time $k^{1+\alpha} n^{O(1)}$  and
 outputs $\tx \in \R^N$ that satisfies \eqref{eqn:approx} (where $\hx = \DHT(x)$)
 for some absolute constant $c > 0$.
\end{thm}
\addtocounter{thm}{-1}
\let\thethm\oldthm

\subsection{The sparse recovery algorithm} \label{sec:algo}

\oldLinespread
\begin{figure}

\begin{mdframed}

\begin{codebox}
\Procname{$\proc{Search}(j \in \F_2^r, t \in [D], s \in \N)$}
\li \For $b \gets 1$ \To $n$
\li \Do
\If $|y^{s, t, b}(j)| \geq |y^{s, t, 0}(j)|/2$
\li \Then
$u_b \gets 1$.
\li \Else
\li $u_b \gets 0$.
\End \End
\li \Return $(u_1, \ldots, u_n)$.
\End
\end{codebox}

\begin{codebox}
\Procname{$\proc{Estimate}(t \in [D], s \in \N)$}
\li Initialize $S \subseteq \F_2^n$ as $S \gets \emptyset$.
\li Initialize $\Delta^{s, t} \in \R^{N}$ as $\Delta^{s, t} \gets 0$.
\li Let $T \subseteq \F_2^r$ be the set of coordinate positions corresponding to \label{ln:defT} \\ the
largest $2k$ entries of $y^{s, t, 0}$.
\li \For $j \in T$
\Do
\li $u \gets \proc{Search}(j, t, s)$.
\li \If $h(u, t) \in T$
\li \Then $S \gets S \cup \{ u \}$.
\li $\Delta^{s, t}(u) \gets y^{s, t, 0}(h(u, t))$.
\End \End
\li \Return $\Delta^{s, t}$.
\end{codebox}


\begin{codebox}
\Procname{$\proc{Recover}(y \in \R^{2^r D n}, s_0 \in \N)$} 
\li $s \gets 0$.
\li Let $B_1, \ldots, B_n \in \zo^{1 \times N}$ be the rows of the bit selection matrix $B$.
\li Initialize $x^0 \in \R^{N}$ as $x^0 \gets 0$.
\li \For $(t, b, j) \in [D] \times \{0, \ldots, n\} \times \F_2^r$
\Do \li $y^{0, t, b}(j)  \gets y(j, t, b)$.
\End
\li \Repeat
\Do
\li \For $t \in [D]$ \label{algo:fotT}
\Do
\li $y^{s, t, 0} \gets M^t \cdot (x-x^s) \in \R^{2^r}$. \label{algo:yst0}
\li \For $b \in [n]$
\Do
\li $y^{s, t, b} \gets (M^t \otimes B_b) \cdot  (x-x^s) \in \R^{2^r}$. \label{algo:ystb}
\End
\li $\Delta^{s, t} \gets \proc{Estimate}(t, s)$.
\End
\li Let $t_0$ be the choice of $t \in [D]$ that minimizes \label{algo:t0} 
$\| \Mz x - \Mz (x^s + \Delta^{s, t}) \|_1$. 
\li $x^{s+1} \gets H_k(x^s + \Delta^{s, t_0})$.
\li $s \gets s + 1$.
\li \End \Until $s = s_0$.
\li Set $x^\ast$ to be the choice of $x^s$ (for $s = 0, \ldots, s_0$) that \label{algo:xstar} 
minimizes $\| \Mz x - \Mz x^s \|_1$. 
\End 
\li \label{algo:return} \Return $x^{\ast}$.
\end{codebox}

\end{mdframed}

\caption{Pseudo-code for the reconstruction algorithm $\proc{Recover}(y, s_0)$, where
$y$ is the sketch $M' x$ and $s_0$ specifies the desired number of iterations. It suffices to set
$s_0 = n^{O(1)}$ according to the bit length of $x$. Notation is explained in Section~\ref{sec:algo}.}
\label{fig:code}
\end{figure}
\newLinespread


The claim of Theorem~\ref{thm:main} is shown using the algorithm presented in Figure~\ref{fig:code}.
%
In this algorithm, $M'$ is the $D2^r(n+1) \times N$ formed by stacking
$M$ on top of $M \otimes B$ and the algorithm is given $y := M' x$ for a vector $x \in \R^N$ to be approximated.
For each $t \in [D]$, we define the $2^r \times N$ matrix $M^t$ to
be the adjacency matrix of the bipartite graph $G^t$ associated with $h_t$ (according
to Definition~\ref{def:truthtable}). 
For $b \in [n]$ we let $B_b \in \zo^{1 \times N}$ be the $b$th row of $B$.
We assume that the entries of $y$ are indexed by the set $\F_2^r \times [D] \times \{0, \ldots, n\}$
where the entry $(a, t, 0)$ corresponds to the inner product defined by the $a$th row of $M^t$ and 
the entry $(a, t, b)$ (for $b \neq 0$) corresponds to the $a$th row of $M^t \otimes B_b$.
Since each entry of $x$ is $n^{O(1)}$ bits long, by using appropriate scaling 
we can without loss of generality assume that
$x$ has integer entries in range $[-L, +L]$ for some $L$ such that $\log L = n^{O(1)}$, and
the algorithm's output can be rounded to the nearest integer in each coordinate
so as to make sure that the final output is integral.



\newcommand{\first}{\mathsf{First}}
\newcommand{\dom}{\mathsf{dom}}


The main ingredient of the analysis is the following lemma which is proved in Appendix~\ref{sec:l:main:proof}.

\begin{lem}
\label{l:main}
For every constant $\gamma > 0$, there is an $\eps_0$ only depending on $\gamma$
such that if $\eps \leq \eps_0$ the following holds.
Suppose that for some $s$,  
\begin{eqn}
\|x-x^s\|_1 > C \|x - H_k(x)\|_1  
\end{eqn}
for  $C=1/\epsilon$. Then,
there is a $t \in [D]$ such that
\begin{eqn}
\|x-(x^s+\Delta^{s,t})\|_1 \le \gamma \|x-x^s\|_1. 
\end{eqn}
\end{lem}

The above lemma can be used, in conjunction with the fact that $M$ satisfies RIP-1, to show that
if $\eps$ is a sufficiently small constant, we can ensure exponential progress
$\|x-x^{s+1}\|_1 \le \|x-x^s\|_1/2$ (shown in Corollary~\ref{coro:main}) 
until the approximation error $\|x-x^s\|_1$ 
reaches the desired level of $C \|x-H_k(x)\|_1$ (after the final truncation).
Then it easily follows that $s_0 = \log(NL)+O(1) = n^{O(1)}$ iterations would suffice
to deduce Theorem~\ref{thm:main}.
Formal proof of Theorem~\ref{thm:main} appears in Section~\ref{app:main}. 

\subsection{Analysis of the running time} \label{sec:main:time}

In order to analyze the running time of the procedure $\proc{Recovery}$,
we first observe that all the estimates $x^0, \ldots, x^{s_0}$ are $k$-sparse vectors and
can be represented in time $O(k (\log n+\log L))$ by only listing the positions and values
of their non-zero entries. In this section we assume that all sparse $N$-dimensional vectors
are represented in such a way.
We observe the following. 

\begin{prop} \label{prop:sparseMult}
Let $w \in \R^N$ be $k$-sparse. Then, for any $t \in [D]$, the products $(M^t \otimes B) \cdot w$
and $M^t w$ can be computed in time $n^{O(1)} (k + 2^r) \ell$, assuming each entry of $w$ is represented within
$\ell$ bits of precision.
\end{prop}

\begin{Proof}
Let $B_1, \ldots, B_n \in \zo^{1 \times N}$ be the rows of the bit selection matrix $B$.
Observe that each column of $M^t$ is entirely zero except for a single $1$ (this is because
$M^t$ represents the truth table of the function $h_t$).  The product $M^t \cdot w$
is simply the addition of at most $k$ such $1$-sparse vectors, and thus, is itself $k$-sparse.
The nonzero entries of $M^t \cdot w$ along with their values can thus be computed by
querying the function $h_t$ in up to $k$ points (corresponding to the support of $w$) followed
by $k$ real additions. Since $h_t$ can be computed in polynomial time
in $n$, we see that $M^t \cdot w$ can be computed in time $n^{O(1)} (k + 2^r) \ell$ (we may assume
that the product is represented trivially as an array of length $2^r$ and thus it takes $2^r$
additional operations to initialize the result vector).
The claim then follows once we observe that for every $b \in [n]$, the matrix
$M^t \otimes B_b$ is even more sparse than $M^t$.
\end{Proof}

Observe that the procedure $\proc{Search}$ needs $O(n)$ operations. 
In procedure $\proc{Estimate}$, identifying $T$ takes $O(2^r)$ time,
and the loop runs for $2k$ iterations, each taking $O(nk)$ time.
In procedure $\proc{Recover}$, we note that for all $t$, $M^t x$ as well as $(M^t \otimes B_b)x$ for all $b \in [n]$ is given as a part
of $y$ at the input. Moreover, all the vectors $x^s$ and
$\Delta^{s, t}$ are $O(k)$-sparse. Thus in light of Proposition~\ref{prop:sparseMult}
and noting that $L = 2^{n^{O(1)}}$ and 
the fact that $h$ is a lossless condenser (which implies $2^r = \Omega(k)$),
we see that computation of each product in Lines \ref{algo:yst0} and \ref{algo:ystb}
of procedure $\proc{Recover}$ takes time $n^{O(1)} 2^r$.
Since the for loop runs for $D$ iterations and so is the number of iterations,
the running time of the loop is $n^{O(1)} D 2^r$.
With a similar reasoning, the computation in Line~\ref{algo:t0} takes time $n^{O(1)} D^2 2^r$.
Similarly, computation of the product in Line~\ref{algo:xstar} of procedure $\proc{Recover}$
takes time $n^{O(1)} D 2^r s_0$. Altogether, recalling that $s_0 = \log(NL) + O(1) = n^{O(1)}$, the total
running time is $n^{O(1)} D^2 2^r$.

\subsection{Proof of Theorem~\ref{thm:main}} \label{app:main}

Theorem~\ref{thm:main} is proved using the algorithm presented in
Figure~\ref{fig:code} and discussed in Section~\ref{sec:algo}. 
We aim to set up the algorithm so that it outputs
a $k$-sparse estimate $\tx \in \R^N$ satisfying \eqref{eqn:approx}.
Instead of achieving this goal, we first consider the following slightly different estimate
\begin{equation} \label{eqn:approx:eta}
\| \tx - x \|_1 \leq C \| x - H_k(x) \|_1 + \nu \| x \|_1,
\end{equation}
for an absolutate constant $C > 0$, where $\nu > 0$ is an arbitrarily small ``relative error'' parameter.
Let us show that this alternative guarantee implies \eqref{eqn:approx}, after
rounding the estimate obtained by the procedure $\proc{Recover}$ to the nearest integer vector.
Recall that without loss of generality (by using appropriate scaling), we can assume that
$x$ has integer coordinates in range $[-L, +L]$, for some $L$ satisfying $\log L = n^{O(1)}$.

\begin{prop} \label{prop:round}
Let $x \in \R^N$ be an integer vector with integer coordinates in range $[-L, +L]$,
and $\tx \in \R^N$ be so that \eqref{eqn:approx:eta} holds
for some $\nu \leq 1/(4NL)$. Let $\tx'$ be the vector obtained by rounding each
entry of $\tx$ to the nearest integer. Then, $\tx'$ satisfies 
\[
\| \tx' - x \|_1 \leq (3C+1/2) \cdot \| x - H_k(x) \|_1.
\]
\end{prop}

\begin{Proof}
If $x = 0$, there is nothing to show. Thus we consider two cases.

\paragraph{Case 1: $\|x-H_k(x)\|_1 = 0$.}

In this case, since $\|x\|_1 \leq NL$, we see that
$
\|\tx-x\|_1 \leq 1/4.
$
Therefore, rounding $\tx$ to the nearest integer vector would exactly recover $x$.

\paragraph{Case 2: $\|x-H_k(x)\|_1 > 0$.}

Since $x$ is an integer vector, we have $\|x-H_k(x)\|_1 \geq 1$. Therefore, again
noting that $\|x\|_1 \leq NL$, from \eqref{eqn:approx:eta} we see that
\[
\| \tx - x \|_1 \leq (C+1/4) \cdot \| x - H_k(x) \|_1.
\]
Therefore, by an averaging argument,
the number of the coordinate positions at which $\tx$ is different from $x$ by $1/2$ or more
is at most $2(C+1/4) \cdot \|x-H_k(x)\|_1$. 
Since rounding can only cause error at such positions, and by at most
$1$ per coordinate, the added error caused by rounding would be at most $2(C+1/4) \cdot \|x-H_k(x)\|_1$,
and the claim follows.
\end{Proof}

In light of Proposition~\ref{prop:round} above, in the sequel we focus on achieving
\eqref{eqn:approx:eta}, for a general $\nu$, and will finally choose $\nu := 1/(4NL)$ so that
using Proposition~\ref{prop:round} we can attain the original estimate in \eqref{eqn:approx}.
We remark that Proposition~\ref{prop:round} is the only place in the proof that assumes
finite precision for $x$ and we do not need such an assumption for achieving \eqref{eqn:approx:eta}.

A key ingredient of the analysis is the following result (Lemma~\ref{lem:BGIKS} below) shown in \cite{ref:BGIKS08}. Before presenting
the result, we define the following notation.

\begin{defn}
Let $w = (w_1, \ldots, w_N) \in \R^N$ be any vector and $G$ be any bipartite graph
with left vertex set $[N]$ and edge set $E$. Then,  $\first(G, w)$ denotes the following subset of edges:
\[
\first(G, w) := \{e = (i, j) \in E \mid (\forall e'=(i', j) \in E)\colon (|w_i| > |w_{i'}|) \lor
(|w_i| = |w_{i'}| \land i' > i) \}. 
\]
\end{defn}

\begin{lem} \cite{ref:BGIKS08} \label{lem:BGIKS}
Let $G$ be a $(k', \eps)$-unbalanced expander graph with left vertex set $[N]$ and edge set $E$.
Then, for any $k'$-sparse vector $w = (w_1, \ldots, w_N) \in \R^N$,  we have
\[
\sum_{(i, j) \in E \setminus \first(G, w)} |w_i| \leq \eps \sum_{(i, j) \in E} |w_i|.
\]
\end{lem}

Intuitively, for every 
right vertex in $G$, $\first(G, w)$ picks exactly one edge connecting the vertex to the left neighbor
at which $w$ has the highest magnitude (with ties broken in a consistent way), and
Lemma \ref{lem:BGIKS} shows that these edges pick up most of the $\ell_1$ mass of $w$.

We apply Lemma~\ref{lem:BGIKS} to the graph $\Gz$ that we set to be the graph
associated with the function $h$. 
Note that this graph is a $(4k, \eps)$-unbalanced expander by Lemma~\ref{lem:cond:expander}.
This means that for every
$(4k)$-sparse vector $w$ and letting $E$ denote the edge set of $\Gz$, we have
\[
\sum_{(i, j) \in E \setminus \first(G, w)} |w_i| \leq \eps \sum_{(i, j) \in E} |w_i| = \eps D \|w\|_1,
\]
where the last equality uses the fact that $\Gz$ is $D$-regular from left. By an averaging argument,
and noting that $\Gz$ is obtained by taking the union of the edges of graphs $G^1, \ldots, G^D$
(each of which being $1$-regular from left), we get that for some $t(G, w) \in [D]$,
\begin{equation} \label{eqn:first}
\sum_{(i, j) \in E^{t(G, w)} \setminus \first(G, w)} |w_i| \leq \eps \|w\|_1,
\end{equation}
where $E^{t(G, w)}$ denotes the edge set of $G^{t(G, w)}$.

Our goal will be to show that the algorithm converges exponentially to the near-optimal solution. In particular, in the following 
we show that if the algorithm is still ``far'' from the optimal solution on the $s$th iteration, it obtains an improved approximation
for the next iteration. This is made precise in Lemma~\ref{l:main}, which we recall below.

\let\oldthmA\thethm
\renewcommand{\thethm}{\ref{l:main}}
\begin{lem} (restated)
\label{l:main:repeat}
For every constant $\gamma > 0$, there is an $\eps_0$ only depending on $\gamma$
such that if $\eps \leq \eps_0$ the following holds.
Suppose that for some $s$,  
\begin{equation} \label{eqn:lmain:assumption}
\|x-x^s\|_1 > C \|x - H_k(x)\|_1  
\end{equation}
for  $C=1/\epsilon$. Then,
there is a $t \in [D]$ such that
\begin{equation} \label{eqn:improve}
 \|x-(x^s+\Delta^{s,t})\|_1 \le \gamma \|x-x^s\|_1. 
\end{equation}
\end{lem}
\let\thethm\oldthmA
\addtocounter{thm}{-1}
%
The proof of Lemma~\ref{l:main} is deferred to Section~\ref{sec:l:main:proof}.

\begin{prop} \label{prop:xpxpp}
Suppose $x', x'' \in \R^{N}$ are $(3k)$-sparse and satisfy
\[\|\Mz (x- x')\|_1 \leq \|\Mz (x- x'')\|_1.\]
Then,
\[\|x- x'\|_1 \leq \Big(1+\frac{3+C_0 \eps}{1-C_0 \eps}\Big) \|x-H_k(x)\|_1 + \frac{1+C_0 \eps}{1-C_0 \eps} \cdot \|x-x''\|_1\]
where $C_0$ is the constant in Theorem~\ref{thm:exp:RIP}. In particular when $C_0 \eps \leq 1/2$, we have
\[\|x- x'\|_1 \leq 8 \|x-H_k(x)\|_1 +  3 \|x-x''\|_1.\]
\end{prop}
\begin{Proof}
\begin{align}
\|x-x'\|_1 &\leq \|x-H_k(x)\|_1+\|H_k(x)-x'\|_1 \label{eqn:xpxpp:a} \\
&\leq \|x-H_k(x)\|_1+\frac{\|\Mz H_k(x)-\Mz x'\|_1}{D(1-C_0 \eps)}  \label{eqn:xpxpp:b} \\
&\leq \|x-H_k(x)\|_1+\frac{\|\Mz x-\Mz x'\|_1 + \|\Mz (x - H_k(x))\|_1}{D(1-C_0 \eps)}  \label{eqn:xpxpp:c} \\
&\leq \|x-H_k(x)\|_1+\frac{\|\Mz x-\Mz x''\|_1 + \|\Mz (x - H_k(x))\|_1}{D(1-C_0 \eps)}  \label{eqn:xpxpp:d} \\
&\leq \|x-H_k(x)\|_1+\frac{\|\Mz H_k(x)-\Mz x''\|_1 + 2 \|\Mz (x - H_k(x))\|_1}{D(1-C_0 \eps)}  \label{eqn:xpxpp:e} \\
&\leq \|x-H_k(x)\|_1+\frac{(1+C_0 \eps)\|H_k(x)-x''\|_1 + 2 \|x - H_k(x)\|_1}{(1-C_0 \eps)}  \label{eqn:xpxpp:f} \\
&\leq \|x-H_k(x)\|_1+\frac{(1+C_0 \eps)\|x-x''\|_1 + (3+C_0\eps) \|x - H_k(x)\|_1}{(1-C_0 \eps)}  \label{eqn:xpxpp:g} \\
&\leq \Big(1+\frac{3+C_0 \eps}{1-C_0 \eps}\Big) \|x-H_k(x)\|_1 + \frac{1+C_0 \eps}{1-C_0 \eps} \cdot \|x-x''\|_1 \label{eqn:xpxpp:h} 
\end{align}
In the above, \eqref{eqn:xpxpp:a}, \eqref{eqn:xpxpp:c}, \eqref{eqn:xpxpp:e}, and \eqref{eqn:xpxpp:g} use the triangle inequality
(after adding and subtracting $H_k(x)$, $\Mz x$, $\Mz H_k(x)$, and $x$ inside the norms, respectively); 
\eqref{eqn:xpxpp:b} and \eqref{eqn:xpxpp:f} use RIP-1 of the matrix $\Mz$ (seeing that $x'$, $x''$, and $H_k(x)$ are sufficiently sparse); 
\eqref{eqn:xpxpp:d} uses the assumption that $\|\Mz (x- x')\|_1 \leq \|\Mz (x- x'')\|_1$; \eqref{eqn:xpxpp:f} also
uses the fact that all columns of $\Mz$ have Hamming weight $D$ and thus the matrix cannot increase
the $\ell_1$ norm of any vector by more than a factor $D$.
\end{Proof}

%
%
%
The following corollary is implied by Lemma~\ref{l:main}.

\begin{coro}  \label{coro:main}
For every constant $\gamma_0 > 0$, there is an $\eps_0$ only depending on $\gamma_0$
such that if $\eps \leq \eps_0$ the following holds.
Assume condition \eqref{eqn:lmain:assumption} of Lemma~\ref{l:main} holds. Then,
\[\|x-x^{s+1}\|_1 \le \gamma_0 \|x-x^s\|_1.\]
\end{coro}
\begin{Proof}
Let $t_0 \in [D]$ be the value computed in Line~\ref{algo:t0} of the procedure $\proc{Recover}$,
and $t \in [D]$ be the value guaranteed to exist by Lemma~\ref{l:main}.
From the fact that the algorithm picks $t_0$ to be the minimizer
of the quantity $\|\Mz x-\Mz (x^s+\Delta^{s,{t}})\|_1$ for all $t\in [D]$,
we have that 
\[\|\Mz x-\Mz (x^s+\Delta^{s,{t_0}})\|_1 \leq \|\Mz x-\Mz (x^s+\Delta^{s,{t}})\|_1. \]
Note that $x^s$ is $k$-sparse and $\Delta^{s,t_0}$ and $\Delta^{s,t}$ are $(2k)$-sparse.
Thus we can apply Proposition~\ref{prop:xpxpp} and deduce that
\[
\|x-(x^s + \Delta^{s,t_0})\|_1 \leq 
\Big(1+\frac{3+C_0 \eps}{1-C_0 \eps}\Big) \|x-H_k(x)\|_1 + \frac{1+C_0 \eps}{1-C_0 \eps} \cdot \|x-(x^s + \Delta^{s,t})\|_1.
\]
Plugging in the bound implied by Lemma~\ref{l:main} and \eqref{eqn:lmain:assumption} in the above inequality
we get
\begin{equation} \label{eqn:nextErr}
\|x-(x^s + \Delta^{s,t_0})\|_1 \leq \gamma' \|x-x^s\|_1,
\end{equation}
where we have defined
\[
\gamma' := \eps \Big(1+\frac{3+C_0 \eps}{1-C_0 \eps}\Big) + \frac{\gamma (1+C_0 \eps)}{(1-C_0 \eps)}.
\]
Now, we can write
\begin{eqnarray}
  \|x-x^{s+1}\|_1 & = &  \|x - H_k(x^s+\Delta^{t_0,s})\|_1 \nonumber \\
  &  \le &  \|x-(x^s+\Delta^{t_0,s}) \|_1 + \|x^s+\Delta^{t_0,s} - H_k(x^s+\Delta^{t_0,s})\|_1 \label{eqn:nextErr:a} \\
  & \le &   \|x-(x^s+\Delta^{t_0,s}) \|_1  + \|x^s+\Delta^{t_0,s}-H_k(x)\|_1 \label{eqn:nextErr:b} \\
  & \le & 2 \|x-(x^s+\Delta^{t_0,s}) \|_1 + \|x-H_k(x)\|_1 \label{eqn:nextErr:c} \\
  & \le &  (2\gamma'+\epsilon) \|x-x^s\|_1. \label{eqn:nextErr:d}
\end{eqnarray}  
In the above, \eqref{eqn:nextErr:a} and \eqref{eqn:nextErr:c} use the triangle inequality (after adding and subtracting
$x^s+\Delta^{t_0,s}$ inside the norm; \eqref{eqn:nextErr:b} uses the fact that $H_k(x)$ and $H_k(x^s+\Delta^{t_0,s})$
are both $k$-sparse by definition and  $H_k(x^s+\Delta^{t_0,s})$ is the best approximator of
$x^s+\Delta^{t_0,s}$ among all $k$-sparse vectors; and \eqref{eqn:nextErr:d} uses 
\eqref{eqn:lmain:assumption} and \eqref{eqn:nextErr}. Finally, note that we can choose $\gamma$ and
$\eps$ small enough so that $2\gamma'+\eps \leq \gamma_0$.
\end{Proof}

For the rest of the analysis, we set $\eps$ a small enough constant so that
\begin{enumerate}
\item $C_0 \eps \leq 1/2$, where $C_0$ is the constant in Theorem~\ref{thm:exp:RIP}.
\item $\gamma_0 = 1/2$, where $\gamma_0$ is the constant in Corollary~\ref{coro:main}.
\end{enumerate}

Observe that for the first iteration of the algorithm, the estimation error is
$\|x-x^0\|_1 = \|x\|_1$.
By repeatedly applying the exponential decrease guaranteed by Corollary~\ref{coro:main}, 
we see that as long as $s_0 \geq \log(3/\nu)$,
we can ensure that at some stage $s \leq s_0$ we attain
\[
\|x-x^s\|_1 \leq C \|x-H_k(x)\|_1 + (\nu/3) \|x\|_1.
\]
Let $x^\ast$ be the estimate computed
in the end of procedure $\proc{Recover}$. 
Recall that both $x^\ast$ and $x^s$ are $k$-sparse vectors. Thus, by 
Proposition~\ref{prop:xpxpp} we see that
\[\|x- x^\ast\|_1 \leq 8 \|x-H_k(x)\|_1 +  3 \|x-x^s\|_1 \leq (3C+8) \cdot \|x-H_k(x)\|_1 + \nu \|x\|_1. \]
Finally, as discussed in the beginning of the analysis, 
by choosing $\nu := 1/(4NL)$ (and thus, $s_0 = \log(NL) + O(1) = n^{O(1)}$) and using Proposition~\ref{prop:round},
the analysis (and proof of Theorem~\ref{thm:main}) is complete. \qed

\subsection{Proof of Lemma~\ref{l:main}} \label{sec:l:main:proof}

We start with some notation. Let $U$ denote the set of the $k$ largest (in magnitude) coefficients of $x$, and let $V$ be the support of $x^s$. Furthermore, we set $W=U \cup V$ and
$z=x -x^s$. That is, $z$ is the vector representing the current estimation error vector. Note that $|W| \le 2k$ and that
$H_k(x) = x_U$. With a slight abuse of notation, we will use the sets $\zo^n$ and $[N]$ interchangeably (for example
in order to index coordinate positions of $z$) and implicitly assume the natural $n$-bit representation of integers in $[N]$
in doing so.


We first apply the result of Lemma~\ref{lem:BGIKS} to the vector $z_W$ so as to conclude that, 
for some $t \in [D]$, according to \eqref{eqn:first} we have
\begin{equation} \label{eqn:first:b}
\sum_{\substack{(i, j) \in E^{t} \setminus \first(G, z_W) \\ i \in W}} |z(i)|
\leq \eps \|z_W\|_1.
\end{equation}
We fix one particular such choice of $t$ for the rest of the proof. 
Define the set
\[
D := \{ i \in [N]\mid (i, h_t(i)) \in \first(G, z_W) \}.
\] 
Intuitively, $\first(G, z_W)$ resolves collisions incurred by $h_t$ by picking, for each hash output,
only the pre-image with the largest magnitude (according to $z_W$). In other words,  
$\first(G, z_W)$ induces a partial function from $[N]$ to $\F_2^r$ that is one-to-one,
and $D$ defines the domain of this partial function.  Using \eqref{eqn:first:b}, we thus have
\begin{equation} \label{eqn:first:c}
\|z_{W \setminus D}\|_1 = \sum_{i \in W \setminus D} |z(i)| \leq \eps \|z_W\|_1 \leq \eps \|z\|_1.
\end{equation}
Define, for any $i \in [N]$,
\begin{equation}
\label{eqn:defDi}
d_i:=\left\|z_{   h_t^{-1}(h_t(i)) \setminus \{i\}   } \right\|_1 = \sum_{i' \in h_t^{-1}(h_t(i)) \setminus \{i\}} |z(i)|.
\end{equation}
Intuitively, with respect to the hash function $h_t$, the quantity $d_i$
collects all the mass from elsewhere that fall into the same bin as $i$.

Our aim is to show that $\Delta^{s,t}$ which is the estimate on the error vector produced 
by the algorithm recovers ``most'' of the coefficients in $z_W$, 
and is therefore ``close'' to the actual error vector $z$.

Our analysis will focus on coefficients in $z_W$ that are ``good'' in the following sense. 
Formally,  we define the set of {\em good} coefficients $\Good$ to contain coefficients $i$ such that:
\begin{enumerate}
\item $i \in W \cap D$, and,
\item $d_i < \delta |z(i)|$, for some small parameter $\delta \leq 1/4$ to be determined later. 
\end{enumerate} 
Intuitively, $\Good$ is the set of coefficients $i$ that ``dominate'' their bucket mass $y(h_t(i))$. 
Thus applying the binary search on any such bucket (i.e., procedure
$\proc{Search}(h_t(i), t, s)$) will return the correct value $i$ (note that
the above definition implies that for any $i \in \Good$, we must have
$y^{s,t,0}(h_t(i)) \neq 0$, and thus the binary search would not degenerate).
More formally, we have the following.
\begin{prop} \label{prop:identifyGood}
For any $i \in \Good$, the procedure 
$\proc{Search}(h_t(i), t, s)$ returns $i$.
\end{prop}
\begin{Proof}
Consider the sequence $(u_1, \ldots, u_n)$ produced by the procedure $\proc{Search}$ and any $b \in [n]$.
 Recall that $y^{s,t,0}=M^t z$ and for each $b \in [n]$, $y^{s,t,b}=(M^t \otimes B_b) \cdot z$.
 Let $j := h_t(i)$. Since $i \in \Good$, we have $d_i  < \delta |z(i)| \leq |z(i)|/2$. Therefore,
\begin{equation} \label{eqn:zvsy}
 |z(i)|(1-\delta) < |y^{s,t,0}(j)| < z(i) (1+\delta).
\end{equation} 
Let $b \in [n]$ and $v \in \zo$ be the $b$th bit in the 
$n$-bit representation of $i$. Let $S$ be the set of those elements in $h_t^{-1}(j) \subseteq \zo^n$ 
  whose $b$th bit is equal to $1$. Note that $i \in S$ iff $v=1$. Recall that
  \[
  y^{s,t,b}(j) = \sum_{i' \in S} z(i').
  \]
  Whenever $i \notin S$, we get
  \[
  |y^{s,t,b}(j)| = \left|\sum_{i' \in S } z(i') \right| 
  \leq \sum_{i' \in h_t^{-1}(j) \setminus \{i\} } |z(i')| = d_i < \delta |z(i)| < \frac{\delta |y^{s,t,0}(j)| }{1-\delta},
  \]
  according to the definition of $d_i$ and \eqref{eqn:zvsy}.
  On the other hand, when $i \in S$, we have
  \[
  |y^{s,t,b}(j)| \geq |z(i)| - \left|\sum_{i' \in S \setminus \{i\} } z(i')\right| \geq |z(i)|-d_i >
  |z(i)|(1-\delta) > \frac{(1-\delta)|y^{s,t,0}(j)| }{1+\delta},
  \]
  again according to the definition of $d_i$ and \eqref{eqn:zvsy}. Thus, the procedure $\proc{Search}$ 
  will be able to distinguish between the two cases $i \in S$ and $i \notin S$ (equivalently, $v=1$ and $v=0$)
  and correctly set $u_b = v$ provided that
  \[
  \frac{\delta}{1-\delta} < \frac{1}{2}
  \]
  and
  \[
  \frac{1-\delta}{1+\delta} \geq \frac{1}{2}
  \]
  which is true according to the choice $\delta \leq 1/4$.
\end{Proof}

By rewriting assumption \eqref{eqn:lmain:assumption} of the lemma, we know that
\[
\|z\|_1  \geq C \| x_{\overline{U}} \|_1,
\]
and thus,
\begin{equation} \label{eqn:zWbar}
\|z_{\overline{W}}\|_1 = \|x_{\overline{W}} \|_1 \leq \|x_{\overline{U}}\|_1 \leq \|z\|_1/C = \eps \|z\|_1,
\end{equation}
where the first equality uses the fact that $x$ and $z=x-x^s$ agree outside $V=\supp(x^s)$ (and thus, outside $W$)
and we also recall that $W \subseteq U$.

Observe that for each $i, i' \in D$ such that $i \neq i'$, we have $h_t(i) \neq h_t(i')$
(since $\first(G, z_W)$ picks exactly one edge adjacent to the right vertex $h_t(i)$, namely $(i, h_t(i))$,
and exactly one adjacent to $h_t(i')$, namely $(i', h_t(i'))$).
In other words for each $i \in D$, the set 
$h_t^{-1}(h_t(i))$ cannot contain any element of $D$ other than $i$. Therefore, we have
\begin{equation} \label{eqn:sumDi}
\sum_{i \in W  \cap D} d_i \le \| z_{\overline{D}} \|_1 \leq \|z_{W \setminus D}\|_1 + \| z_{\overline{W}} \|_1
\leq 2\eps \|z\|_1,
\end{equation}
where for the last inequality we have used \eqref{eqn:first:c}~and~\eqref{eqn:zWbar}.

Now we show that a substantial portion of the $\ell_1$ mass of $z$ is collected by the set of good indices $\Good$. 
\begin{lem}
\label{l:Gz}
$ \sum_{i \in \Good} |z(i)| \ge (1-2\epsilon (1+1/\delta))\|z\|_1$.
\end{lem}
\begin{Proof}
We will upper bound $\sum_{i \notin \Good} |z(i)|$, and in order to do so, decompose this sum into three components bounded as follows:
\begin{itemize}
\item $\sum_{i \notin W} |z(i)| \le \epsilon \|z\|_1$ (according to \eqref{eqn:zWbar})
\item $\sum_{i \in W \setminus D} |z(i)| \le \epsilon \|z\|_1$ (according to \eqref{eqn:first:c})
\item $\sum_{(W \cap D)  \setminus \Good} |z(i)| \le 2 \epsilon/\delta \|z\|_1$. In order to verify this claim,
observe that from the definition of $\Good$, every $i \notin \Good$ satisfies $|z(i)| \leq d_i/\delta$. Therefore,
the left hand side summation is at most $\sum_{i \in W \cap D} |d_i|/\delta$ and the bound follows
using \eqref{eqn:sumDi}.
\end{itemize}
By adding up the above three partial summations, the claim follows.
\end{Proof}

Lemma~\ref{l:Gz} shows that it suffices to recover most of the coefficients $z_i$ for $i \in \Good$ in order recover most of the $\ell_1$ mass in $z$.
This is guaranteed by the following lemma.

\begin{lem} \label{lem:beta}
There is a $\beta>0$ only depending on $\eps$ and $\delta$ such that
$\beta = O_\delta(\eps)$ and 
\[ \sum_{i \in \Good, h_t(i) \in T} |z(i)| \ge (1-\beta) \|z\|_1, \]
where $T$ is the set define in Line~\ref{ln:defT} of the procedure $\proc{Estimate}$.
\end{lem}

\begin{Proof}
Consider the bin vector $y := y^{s,t,0} = M^t z$. 
From the choice of $T$ as the set picking the largest $2k$ coefficients of $y$,
it follows that for all $j \in T \setminus h_t(\Good)$ and $j' \in h_t(\Good) \setminus T$ (where $h_t(\Good)$ denotes the set
$\{h_t(i) \mid i \in \Good\}$) we have $|y(j)| \ge |y({j'})|$. 
Since $|T|=2k$ and $|h_t(\Good)| \le 2k$ (because $\Good \subseteq W$ which is in turn $(2k)$-sparse), 
it follows that $|T \setminus  h_t(\Good)| \ge |h_t(\Good) \setminus T|$. Therefore,
\[
 \sum_{j \in  h_t(\Good) \setminus  T} |y(j)| \le \sum_{j \in  T \setminus h_t(\Good)} |y(j)|.
\]
 Now, using Lemma~\ref{l:Gz} we can deduce the following.
\begin{eqnarray}
\label{e:con}
 \sum_{j \in  T \setminus h_t(\Good)} |y(j)| \le  \sum_{i \notin {\Good}} |z(i)| \le 2\epsilon (1+1/\delta) \|z\|_1 
\end{eqnarray}
where for the first inequality we note that $y(j) = \sum_{i \in h_t^{-1}(j)}{z(i)}$ and that the sets $h_t^{-1}(j)$ for
various $j$ are disjoint and cannot intersect $\Good$ unless, by definition, $j \in h_t(\Good)$.

Recall that for every $i \in \Good$, by the definition of $\Good$ we have
\begin{equation*} 
(1-\delta) |z(i)| < |y(h_t(i))| < (1+\delta) |z(i)|.
\end{equation*}
Using this, it follows that
\begin{eqnarray}
  \sum_{i \in \Good, h_t(i) \in T} |z(i)| & \ge &  \frac{1}{1+\delta} \sum_{j \in h_t(\Good) \cap T} |y(j)| \nonumber \\
     & \ge & \frac{1}{1+\delta}  \left ( \sum_{j \in h_t(\Good) } |y(j)| -  \sum_{j \in h_t(\Good) \setminus T } |y(j)| \right ) \nonumber \\
      & \ge &  \frac{1}{1+\delta}   \left ( (1-\delta) \sum_{i \in \Good} |z(i)| - \sum_{j \in h_t(\Good)  \setminus  T } |y(j)| \right ) \nonumber \\
     & \ge &  \frac{(1-\delta)(1-2\epsilon(1+1/\delta)) - 2\epsilon (1+1/\delta) }{1+\delta}  \|z\|_1 =: (1-\beta) \|z\|_1,
     \nonumber 
\end{eqnarray}
where the last step follows from Lemma~\ref{l:Gz} and \eqref{e:con}.
\end{Proof}

We are now ready to conclude the proof of Lemma~\ref{l:main}. 
First, observe using Proposition~\ref{prop:identifyGood} that for coordinates $i \in \Good$ such that $h_t(i) \in T$, 
we have $\Delta^{s,t}(i)=y({h_t(i))}$ and
that, since $i \in \Good$,
\begin{equation}
\label{eqn:yVsZ}
z(i)(1-\delta) \leq z(i) - d_i \leq y({h_t(i))} \leq z(i) + d_i \leq z(i)(1+\delta).
\end{equation}

 Therefore, for such choices of $i$, $|\Delta^{s,t}(i) - z(i)|\le \delta |z(i)|$. Thus we have
\begin{align}
\|\Delta^{s,t} -z\|_1 &= \sum_{i \in \Good \cap h_t^{-1}(T) } |\Delta^{s,t}(i) - z(i)| + \sum_{i \notin \Good \cap h_t^{-1}(T) } |\Delta^{s,t}(i) - z(i)| \label{eqn:DeltaZ:a} \\
&\le \delta \|z\|_1 + \sum_{i \notin h_t^{-1}(T) } |\Delta^{s,t}(i) - z(i)| + \sum_{i \in h_t^{-1}(T) \setminus \Good } |\Delta^{s,t}(i) - z(i)| \label{eqn:DeltaZ:b} \\
&= \delta \|z\|_1 + \sum_{i \notin h_t^{-1}(T) } |z(i)| + \sum_{i \in h_t^{-1}(T) \setminus \Good } |\Delta^{s,t}(i) - z(i)| \nonumber \\
&\leq \delta \|z\|_1 + \sum_{i \notin h_t^{-1}(T)} |z(i)| + \sum_{i \in h_t^{-1}(T) \setminus \Good } (|\Delta^{s,t}(i)| + |z(i)|) \nonumber \\
&= \delta \|z\|_1 + \sum_{i \notin h_t^{-1}(T) \cap \Good } |z(i)| + \sum_{i \in h_t^{-1}(T) \setminus \Good } (|\Delta^{s,t}(i)| + |z(i)|) \nonumber \\
&\leq (\delta+\beta) \|z\|_1 + \sum_{i \in h_t^{-1}(T) \setminus \Good } |\Delta^{s,t}(i)| \label{eqn:DeltaZ:c} 
\end{align}

In the above, \eqref{eqn:DeltaZ:b} uses \eqref{eqn:yVsZ} and \eqref{eqn:DeltaZ:c} uses Lemma~\ref{lem:beta}.
Now, for each $i \in h_t^{-1}(T) \setminus \Good$ such that $|\Delta^{s,t}(i)| \neq 0$,
the algorithm by construction sets $\Delta^{s,t}(i) = y^{s,t,0}(h_t(i)) = \sum_{j \in h_t^{-1}(h_t(i))} z(j)$.
Observe that in this case, we must have $h_t^{-1}(h_t(i)) \cap \Good = \emptyset$. This is because
if there is some $i' \in h_t^{-1}(h_t(i)) \cap \Good$, the for loop in procedure $\proc{Estimate}$ upon 
processing the element $h_t(i) = h_t(i')$ in the set $T$ would call $\proc{Search}(h_t(i'), t, s)$ which
would return $i'$ rather than $i$ according to Proposition~\ref{prop:identifyGood} (since $i' \in \Good$),
making the algorithm estimate the value of $\Delta^{s,t}(i)$ and leave $\Delta^{s,t}(i')$ zero.
Therefore,
\[
\sum_{i \in h_t^{-1}(T) \setminus \Good } |\Delta^{s,t}(i)| \leq \sum_{i \notin \Good} |z(i)| \leq \beta\|z\|_1,
\]
the last inequality being true according to Lemma~\ref{l:Gz}. Plugging this result back into \eqref{eqn:DeltaZ:b},
we get that
\[
\|x-(x^s+\Delta^{s,t})\|_1 = \|\Delta^{s,t} -z\|_1 \leq (\delta+2\beta) \|z\|_1 = (\delta+2\beta) \|x-x^s\|_1.
\]
The proof of Lemma~\ref{l:main} is now complete by choosing $\delta$ and $\eps$ (thus $\beta$) small enough constants so that
$\delta  + 2 \beta \leq \gamma$.

\qed

\section{Speeding up the algorithm using randomness}
\label{sec:random}

Although this work focuses on deterministic algorithms for sparse Hadamard transform, in this
section we show that our algorithm in Figure~\ref{fig:code} can be significantly sped up by
using randomness (yet preserving non-adaptivity). 

The main intuition is straightforward:
In the {\bf for} loop of Line~\ref{algo:fotT}, in fact most choices of $t$ turn out to be
equally useful for improving the approximation error of the algorithm. Thus, instead of 
trying all possibilities of $t$, it suffices to just pick one random choice. However, since
the error $\eps$ of the condenser is a constant, the ``success probability'' of
picking a random $t$ has to be amplified. This can be achieved by either 1) Designing
the error of condenser small enough to begin with; or, 2) Picking a few independent
random choices of $t$ and trying each such choice, and then estimating the choice that
leads to the best improvements. It turns out that the former option can be rather 
wasteful in that it may increase the output length of the condenser (an subsequently,
the overall sample complexity and running time) by a substantial factor. In this section,
we pursue the second approach which leads to nearly optimal results.

In this section, we consider a revised algorithm that 

\begin{itemize}

\item Instead of looping over all choices of $t$ in Line~\ref{algo:fotT} of procedure $\proc{Recover}$,
just runs the loop over a few random choices of $t$ .


\item In Line~\ref{algo:return}, instead of minimizing $\|Mx - Mx^s\|_1$,
performs the minimization with respect to a randomly sub-sampled submatrix of $M$
obtained from restriction $M$ to a few random and independent choices of $t$. 
\end{itemize}

The above randomized version of procedure $\proc{Recover}$ is called procedure $\proc{Recover}'$
in the sequel, and is depicted in Figure~\ref{fig:code:random}. The algorithm chooses an integer
parameter $q$ which determines the needed number of samples for $t$. In the algorithm,
we use the notation $M^{\cT}$, where $\cT \subseteq [D]$ is a multi-set, to denote the $|\cT| 2^r \times N$
matrix obtained by stacking matrices $M^t$ for all $t \in \cT$ on top of one another. Note that
the algorithm repeatedly uses fresh samples of $t$ as it proceeds. This eliminates possible
dependencies as the algorithm proceeds and simplifies the analysis.

\oldLinespread
\begin{figure}

\begin{mdframed}

\begin{codebox}
\Procname{$\proc{Recover}'(y, s_0, q)$} 
\li $s \gets 0$.
\li Let $B_1, \ldots, B_n \in \zo^{1 \times N}$ be the rows of the bit selection matrix $B$.
\li Initialize $x^0 \in \R^{N}$ as $x^0 \gets 0$.
\li \For $(t, b, j) \in [D] \times \{0, \ldots, n\} \times \F_2^r$ \label{algo:random:init}
\Do \li $y^{0, t, b}(j)  \gets y(j, t, b)$.
\End
\li \Repeat
\Do
\li Let $\cT^s \subseteq [D]$ be a multiset of $q$ uniformly and independently random elements.
\label{algo:random:Ts}
\li \For $t \in \cT^s$ \label{algo:random:fotT}
\Do
\li $y^{s, t, 0} \gets M^t \cdot (x-x^s) \in \R^{2^r}$. \label{algo:random:yst0}
\li \For $b \in [n]$
\Do
\li $y^{s, t, b} \gets (M^t \otimes B_b) \cdot  (x-x^s) \in \R^{2^r}$. \label{algo:random:ystb}
\End
\li $\Delta^{s, t} \gets \proc{Estimate}(t, s)$.
\End 
\li Let $\cT'^s \subseteq [D]$ be a multiset of $q$ uniformly and independently random elements.
\label{algo:random:Tps}
\li Let $t_0$ be the choice of $t \in \cT^s$ that minimizes \label{algo:random:t0} 
$\| M^{\cT'^s} x - M^{\cT'^s} (x^s + \Delta^{s, t}) \|_1$. 
\li $x^{s+1} \gets H_k(x^s + \Delta^{s, t_0})$.
\li $s \gets s + 1$.
\li \End \Until $s = s_0$.
\li Let $\cT'' \subseteq [D]$ be a multiset of $q$ uniformly and independently random elements.
\label{algo:random:Tpps}
\li Set $x^\ast$ to be the choice of $x^s$ (for $s = 0, \ldots, s_0$) that \label{algo:random:xstar} 
minimizes 
$\| M^{\cT''} x - M^{\cT''} x^s \|_1$. 
\End 
\li \Return $x^{\ast}$.  \label{algo:random:return}
\end{codebox}

\end{mdframed}

\caption{Pseudo-code for the randomized version of the algorithm $\proc{Recover}$. 
The algorithm receives $y$ implicitly and only queries $y$ at a subset of the positions. 
The additional integer parameter $q$ is set up by the analysis. }
\label{fig:code:random}
\end{figure}
\newLinespread

More formally, our goal in this section is to prove the following randomized analogue of Theorem~\ref{thm:main}.
Since the running time of the randomized algorithm may in general be less than the sketch length
($2^r D (n+1)$), we assume that the randomized algorithm receives the sketch implicitly and has 
query access to this vector.

\begin{thm} \label{thm:main:random} (Analogue of Theorem~\ref{thm:main})
There are absolute constants $c > 0$ and $\eps' > 0$ such that the following holds.
Let $k, n$ ($k \leq n$) be positive integer parameters, and suppose there exists a function 
$h\colon \F_2^n \times [D] \to \F_2^r$ (where $r \leq n$) computable in time $f(n)$ (where $f(n) = \Omega(n)$)
which is
an explicit $(\log(4k), \eps')$-lossless condenser. Let $M$ be the adjacency 
matrix of the bipartite graph associated with $h$ and $B$ be the bit-selection
matrix with $n$ rows and $N := 2^n$ columns. 
Then, there is a randomized algorithm that, given $k, n$, parameters $\eta, \nu > 0$, 
and query access to the vectors $M x$ and $(M \otimes B) x$ 
for some $x \in \R^N$ (which is \emph{not} given to the algorithm), 
computes a $k$-sparse estimate $\tx$ such that, with probability at least
$1-\eta$ over the random coin tosses of the algorithm,
 \[
\| \tx - x \|_1 \leq c \cdot \| x - H_k(x) \|_1 + \nu \| x \|_1.
 \]
 Moreover, execution of the algorithm takes 
 $O(2^r \cdot \log(\log(1/\nu)/\eta) \cdot \log(1/\nu) f(n))$ arithmetic operations
 in the worst case.
\end{thm}

Proof of Theorem~\ref{thm:main:random} is deferred to Section~\ref{sec:thm:main:random}.
In the sequel, we instantiate this theorem for use in sparse Hadamard transform application.
Specifically, we consider the additional effect on the running time incurred by the initial sampling stage;
that is, computation
of the input to the algorithm in Figure~\ref{fig:code:random} from the information provided in
$\hx = \DHT(x)$.

First, notice that all the coin tosses of the algorithm in Figure~\ref{fig:code:random}
(namely, the sets $\cT^0, \ldots, \cT^{s_0-1}$, $\cT'^0, \ldots, \cT'^{s_0-1}$, and $\cT''$)
can be performed when the algorithm starts, due to the fact that each random sample $t \in [D]$
is distributed uniformly and independently of the algorithm's input and other random choices.
Therefore, the sampling stage needs to compute $M^t x$ and $(M^t \otimes B) x$
for all the $(2s_0 + 1)q$ random choice of $t$ made by the algorithm.

For $t \in [D]$, let $V_t$ be the $(n-r)$-dimensional subspace of $\F_2^n$ which is
the kernel of the linear function $h_t$. Moreover, let $V_t^\perp$ and $W_t$ respectively
denote the dual and complement of $V_t$ (as in Lemma~\ref{lem:sumVcompute}).
As discussed in the proof of Theorem~\ref{thm:hadamard:sublinear}, for each $t \in [D]$,
we can use Lemma~\ref{lem:sampling} to
compute of $M^t x$ from query access to $\hx = \DHT(x)$ at $O(2^r r)$ points
and using $O(2^r r n)$ arithmetic operations, assuming that a basis for
$V_t^\perp$ and $W_t$ is known. Similarly, $(M^t \otimes B) x$ may be computed using
$O(2^r r n^2)$ operations and by querying $\hx$ at $O(2^r r n)$ points.

Computation of a basis for $V_t^\perp$ and $W_t$ for a given $t$ can in general
be performed\footnote{
For structured transformations it is possible to do better (see \cite{ref:displacement}). This is the case for the specific case of
Leftover Hash Lemma that we will later use in this section. However, we do not attempt to 
optimize this computation since it only incurs an additive poly-logarithmic factor in $N$ which
affects the asymptotic running time only for very small $k$.
} using Gaussian elimination in time $O(n^3)$.
Therefore, the additional time for the pre-processing needed for computation of such bases
for all choices of $t$ picked by the algorithm is $O(q s_0 n^4)$.

Altogether, we see that the pre-processing stage in total takes 
\[O(q s_0 (2^r r + n^2) n^2) = O( \log(\log(1/\nu)/\eta) \cdot \log(1/\nu) \cdot  (2^r r + n^2) n^2 )\] 
arithmetic operations. 

Finally we instantiate the randomized sparse DHT algorithm using Theorem~\ref{thm:main:random},
pre-processing discussed above, and the lossless condensers 
constructed by the Leftover Hash Lemma (Lemma~\ref{lem:leftover}). As for the linear
family of hash functions required by the Leftover Hash Lemma, we use the linear family
$\mathcal{H}_{\mathsf{lin}}$ which is defined in Section~\ref{sec:main:time}. Informally, a hash function
in this family corresponds to an element $\beta$ of the finite field $\F_{2^n}$. Given an
input $x$, the function interprets $x$ as an element of $\F_{2^n}$ and then outputs 
the bit representation of $\beta \cdot x$
truncated to the desired $r$ bits.
We remark that the condenser obtained in this way is computable in
time $f(n) = O(n \log n)$ using the FFT-based multiplication algorithm over $\F_{2^n}$.
Simplicity of this condenser and mild hidden constants in the asymptotics is particularly appealing for practical applications.

Recall that for the Leftover Hash Lemma, we have $2^r = O(k/\eps'^2) = O(k)$, which
is asymptotically optimal.
Using this in the above running time estimate, we see that the final randomized version of the
sparse DHT algorithm performs
\[
O( \log(\log(1/\nu)/\eta) \cdot \log(1/\nu) \cdot  (k \log k + n^2) \cdot n^2 )
\]
arithmetic operations in the worst case to succeed with probability at least $1-\eta$.

Finally, by recalling that an algorithm that computes an estimate satisfying
\eqref{eqn:approx:eta} can be transformed into one satisfying \eqref{eqn:approx}
using Proposition~\ref{prop:round}, we conclude the final result of this section (and Theorem~\ref{thm:general:random:intro}) 
that follows from the above discussion combined with Theorem~\ref{thm:main:random}.

\begin{coro}[Generalization of Theorem~\ref{thm:general:random:intro}] \label{coro:general:random} 
There is a randomized algorithm that, given integers $k, n$ (where $k \leq n$), parameters
$\eta > 0$ and $\nu > 0$, and 
(non-adaptive) query access to any $\hx \in \R^N$ (where $N := 2^n$),
 outputs $\tx \in \R^N$ that, with probability at least $1-\eta$ over the internal random
 coin tosses of the algorithm, satisfies 
 \begin{eqn}
\| \tx - x \|_1 \leq c \| x - H_k(x) \|_1 + \nu \| x \|_1,
 \end{eqn}
 for some absolute constant $c > 0$ and $\hx = \DHT(x)$. 
Moreover, the algorithm performs a worse-case
\[
O( \log(\log(1/\nu)/\eta) \cdot \log(1/\nu) \cdot  (k \log k + n^2) \cdot n^2 )
\]
arithmetic operations\footnote{We remark that the running time estimate counts $O(n)$ operations
for indexing; that is, looking for $\hx(i)$ for an index $i \in [N]$, and one operation
for writing down the result.
} to compute $\tx$. Finally, when each coefficient of $\hx$
takes $O(n)$ bits to represent, the algorithm can be set up to output $\tx$
satisfying
 \begin{eqn}
\| \tx - x \|_1 \leq c \| x - H_k(x) \|_1,
 \end{eqn}
using 
$
O( \log(n/\eta) \cdot (k \log k + n^2) \cdot n^3 )
$
arithmetic operations in the worst case. In particular, when $\eta = 1/n^{O(1)}$ and
$k = \Omega(n^2) = \Omega(\log^2 N)$, the algorithm runs in worse case time $O(k n^3 (\log k) (\log n) ) = \tilde{O}(k (\log N)^3)$. \qed
\end{coro}

\subsection{Proof of Theorem~\ref{thm:main:random}} \label{sec:thm:main:random}

The proof is quite similar to the proof of Theorem~\ref{thm:main}, and therefore, in this section
we describe the necessary modifications to the proof of Theorem~\ref{thm:main} which lead to
the conclusion of Theorem~\ref{thm:main:random}.

\subsubsection{Correctness analysis of the randomized sparse recovery algorithm}

Similar to the proof of Theorem~\ref{thm:main}, our goal is to 
set up the randomized algorithm so that, given arbitrarily small parameters $\nu, \eta > 0$, it outputs
a $k$-sparse estimate $\tx \in \R^N$ that at least with probability $1-\eta$ (over the random coin
tosses of the algorithm) satisfies \eqref{eqn:approx:eta}, recalled below, for an absolute constant $C > 0$:
\begin{equation*} 
\| \tx - x \|_1 \leq C \| x - H_k(x) \|_1 + \nu \| x \|_1,
\end{equation*}
As in the proof of Theorem~\ref{thm:main} and using Proposition~\ref{prop:round},
once we have such a guarantee for some $\nu = \Theta(1/(NL))$, assuming that
$x$ has integer coordinates in range $[-L, +L]$ and by 
rounding the final result vector to the nearest integer vector we get the guarantee in
\eqref{eqn:approx}.



We will also use the following ``error amplification'' result that can be simply proved
using standard concentration results.

\begin{lem} \label{lem:random:amplify}
Suppose $h\colon \F_2^n \times [D] \to \F_2^r$ is a $(\kappa, \eps)$-lossless condenser. 
For any set $S \subseteq \F_2^n$ where
$|S| \leq 2^\kappa$ the following holds. Let $q \in \N$ be a parameter and
$t_1, \ldots, t_{q}$ be drawn uniformly and independently at random. 
Let $h'\colon \F_2^n \times [q] \to \F_2^r$ be defined as
$h'(x, j) := h(x, t_j)$, and $G$ be the bipartite graph associated with $h'$. Let
$T \subseteq \F_2^r$ be the neighborhood of 
the set of left vertices of $G$ defined by $S$. 
Then, with probability at least $1-\exp(-\eps^2 q/4)$
(over the randomness of $t_1, \ldots, t_{q}$),
we have $|T| \geq (1-2\eps) q |S|$. 
\end{lem}

\begin{proof}
Let $G^0$ be the bipartite graph associated with $h$, with $N$ left vertices and $D2^r$
right vertices, and for each $t \in [D]$, denote by $G^t$ the bipartite graph associated with $h_t$,
each having $N$ left vertices and $2^r$ right vertices. Recall that $G^0$ contains the union of the
edge set of $G^1, \ldots, G^D$ (with shared left vertex set $[N]$ and disjoint right vertex sets), 
and that $G$ contains the union of the edge set of $G^{t_1}, \ldots, G^{t_q}$.
Let $T^0$ be the set of right neighbors of $S$ in $G^0$. Similarly, let
$T^t$ ($t \in [D]$) be the set of right neighbors of $S$ in $G^t$. 

Since $h$ is a lossless condenser, we know that $|T^0| \geq (1-\eps) D|S|$.
For $i \in [q]$, let $X_i \in [0, 1]$ be such that $|T^i| = (1-X_i) |S|$, and 
define $X := X_1 + \cdots + X_q$.
By an averaging argument, we see that $\E[X_i] \leq \eps$. 
Moreover, the random variables $X_1, \ldots, X_q$ are independent.
Therefore, by a Chernoff bound, 
\[
\Pr[X > 2\eps q] \leq \exp(-\eps^2 q/4).
\]
The claim follows after observing that $|T| = (q-X)|S|$ (since the graph $G$ is composed
of the union of $G^1, \ldots, G^q$ with disjoint right vertex sets).
\end{proof}

Note that the above lemma requires the set $S$ to be determined and fixed \emph{before} the
random seeds $t_1, \ldots, t_{q}$ are drawn. Thus the lemma makes no claim
about the case where an adversary chooses $S$ based on the outcomes of the
random seeds.


In the sequel, we set the error of the randomness condenser (that we shall denote by $\eps'$) to be $\eps' \leq \eps/2$,
where $\eps$ is the constant from Theorem~\ref{thm:hadamard:sublinear}.

We observe that the result reported in Lemma~\ref{lem:BGIKS} only uses the expansion property
of the underlying bipartite graph with respect to the particular support of the vector $w$.
Thus, assuming that the conclusion of Lemma~\ref{lem:random:amplify} holds
for the set $S$ in the lemma set to be the support of a $k'$-sparse vector $w$ (where in our
case $k' = 4k$), we may
use the conclusion of Lemma~\ref{lem:BGIKS} that, for some $t \in \{t_1, \ldots, t_q\}$,

\begin{equation*}
\sum_{(i, j) \in E^t \setminus \first(G, w)} |w_i| \leq \eps \|w\|_1.
\end{equation*}


Using the above observation, we can deduce an analogue of the result of Lemma~\ref{l:main} for the
randomized case by noting that the result in Lemma~\ref{l:main} holds as long as
the set $W$ in the proof of this lemma satisfies \eqref{eqn:first:b}. 
Since the choice of $W$ only depends on the previous iterations of the algorithm;
that is the algorithm's input and random coin tosses determining $\cT^0, \ldots, \cT^{s-1}$,
we can use Lemma~\ref{lem:random:amplify} to ensure that \eqref{eqn:first:b} holds
with high probability. In other words,
we can rephrase Lemma~\ref{l:main} as follows.

\begin{lem} (Analogue of Lemma~\ref{l:main}) \label{l:random:main}
For every constant $\gamma > 0$, there is an $\eps_0$ and $C>0$ only depending on $\gamma$
such that if $\eps' \leq \eps_0$ the following holds.
Suppose that for some $s$,  
\begin{equation} \label{eqn:lmain:random:assumption}
\|x-x^s\|_1 > C \|x - H_k(x)\|_1.
\end{equation}
Then, with probability at least $1-\exp(-\eps'^2 q/4)$, there is a $t \in \cT^s$ such that
\begin{eqn}
\|x-(x^s+\Delta^{s,t})\|_1 \le \gamma \|x-x^s\|_1. 
\end{eqn}
\qed
\end{lem}

Declare a \emph{bad event} at stage $s$ if we have the condition 
$\|x-x^s\|_1 > C \|x - H_k(x)\|_1$ however the conclusion
of the lemma does not hold because of unfortunate random coin tosses by the algorithm. By a union
bound, we see that the probability that any such bad event happens throughout
the algorithm is at most $s_0 \exp(-\eps'^2 q/4)$.

Next we show an analogue of Proposition~\ref{prop:xpxpp} for the randomized algorithm.

\begin{prop} \label{prop:xpxpp:random}
Let $x', x'' \in \R^{N}$ be fixed $(3k)$-sparse vectors and $\cT$ be 
a multi-set of $q$ elements in $[D]$ chosen uniformly and independently at random.
Moreover, assume
\[\|M^{\cT} (x- x')\|_1 \leq \|M^{\cT} (x- x'')\|_1.\]
Then, with probability at least $1-2\exp(-\eps'^2 q/4)$ over the choice of $\cT$, we have
\[\|x- x'\|_1 \leq \Big(1+\frac{3+C_0 \eps}{1-C_0 \eps}\Big) \|x-H_k(x)\|_1 + \frac{1+C_0 \eps}{1-C_0 \eps} \cdot \|x-x''\|_1\]
where $C_0$ is the constant in Theorem~\ref{thm:exp:RIP}. In particular when $C_0 \eps \leq 1/2$, we have
(with the above-mentioned probability bound)
\[\|x- x'\|_1 \leq 8 \|x-H_k(x)\|_1 +  3 \|x-x''\|_1.\]
\end{prop}

\begin{proof}
Proof is the same as the original proof of Proposition~\ref{prop:xpxpp}. The only difference is observing
that the argument is valid provided that the RIP-1 condition holds
for two particular $(4k)$-sparse vectors $H_k(x) - x''$ and $H_k(x) - x'$ (as used in \eqref{eqn:xpxpp:b} and \eqref{eqn:xpxpp:f}).
On the other hand, the proof of Theorem~\ref{thm:exp:RIP} only uses the expansion property of the
underlying expander graph for the particular support of the sparse vector being considered,
and holds as long as the expansion is satisfied for this particular choice. By applying Lemma~\ref{lem:random:amplify}
twice on the supports of $H_k(x) - x''$ and $H_k(x) - x'$, and taking a union bound, we see that
the required expansion is available with probability at least $1-2\exp(-\eps'^2 q/4)$, and thus the claim follows.
\end{proof}

Using the above tool, we can now show an analogue of Corollary~\ref{coro:main}; that is,

\begin{coro}  \label{coro:random:main}
For every constant $\gamma_0 > 0$, there is an $\eps_0$ only depending on $\gamma_0$
such that if $\eps \leq \eps_0$ the following holds.
Assume condition \eqref{eqn:lmain:random:assumption} of Lemma~\ref{l:random:main} holds. Then,
with probability at least $1-2\exp(-\eps'^2 q/4)$ over the choice of $\cT'^s$, we have
\[\|x-x^{s+1}\|_1 \le \gamma_0 \|x-x^s\|_1.\]
\end{coro}
\begin{Proof}
The proof is essentially the same as the proof of Corollary~\ref{coro:main}.
The only difference is that instead of $\|\Mz x-\Mz (x^s+\Delta^{s,{t}})\|_1$,
the quantity $\|M^{\cT'^s} x-M^{\cT'^s} (x^s+\Delta^{s,{t}})\|_1$ that is used in the
randomized algorithm is considered, and Proposition~\ref{prop:xpxpp:random}
is used instead of Proposition~\ref{prop:xpxpp}. In order to ensure that we can use Proposition~\ref{prop:xpxpp:random},
we use the fact that particular choices of the vectors $x'$ and $x''$ that we instantiate
Proposition~\ref{prop:xpxpp:random} with (respectively, the vectors
$x^s + \Delta^{s, t_0}$ and $x^s + \Delta^{s, t}$ in the proof of Corollary~\ref{coro:main})
only depend on the algorithm's input and random coin tosses determining
$\cT^0, \ldots, \cT^{s}$ and $\cT'^0, \ldots, \cT'^{s-1}$ and not on $\cT'^s$.
\end{Proof}

Again, declare a \emph{bad event} at stage $s$ if we have the condition 
$\|x-x^s\|_1 > C \|x - H_k(x)\|_1$ however the conclusion
of Corollary \ref{coro:random:main} does not hold because of unfortunate 
coin tosses over the choice of $\cT'^s$. 
Same as before, by a union
bound we can see that the probability that any such bad event happens throughout
the algorithm is at most $2 s_0 \exp(-\eps'^2 q/4)$.

Since the initial approximation is $x^0 = 0$ (with error at most $\|x\|$),
assuming $\gamma_0 \leq 1/2$, we have that for some $s \leq \log(1/\nu)$ the condition
\eqref{eqn:approx:eta} is satisfied provided that a bad event does not happen in the
first $s$ iterations. By the above union bounds, this is the case with probability at least 
$1-3 s_0 \exp(-\eps'^2 q/4)$.

Let $x^\ast$ be the estimate computed in Line~\ref{algo:return} of procedure $\proc{Recover}'$.
We can conclude the analysis in a similar way to the proof of Theorem~\ref{thm:main} 
by one final use of Proposition~\ref{prop:xpxpp:random} as follows.
By Proposition~\ref{prop:xpxpp:random}, assuming no bad event ever occurs,
with probability at least $1-2 \exp(-\eps'^2 q/4)$ we see that
\begin{equation}
\label{eqn:approx:eta:final}
\|x- x^\ast\|_1 \leq 8 \|x-H_k(x)\|_1 +  3 \|x-x^s\|_1 \leq C' \cdot \|x-H_k(x)\|_1  + \nu \|x\|_1, 
\end{equation}
where we define $C' := 3C+8$. 

Altogether, by a final union bound we conclude that the desired \eqref{eqn:approx:eta:final}
holds with probability at least $1-\eta$ for some choice of $q = O(\log(s_0/\eta)/\eps'^2) = O(\log(s_0/\eta))$.

\subsubsection{Analysis of the running time of the randomized sparse recovery algorithm}

The analysis of the running time of procedure $\proc{recover}'$ in Figure~\ref{fig:code:random} 
is similar to Section~\ref{sec:main:time}.
As written in Figure~\ref{fig:code:random}, the algorithm may not achieve the promised
running time since the sketch length may itself be larger than the desired running time.
Thus we point out that the sketch is implicitly given to the algorithm as an oracle and
the algorithm queries the sketch as needed throughout its execution. Same holds for
the initialization step in Line~\ref{algo:random:init} of procedure $\proc{Recover}'$,
which need not be performed explicitly by the algorithm.

In order to optimize time, the algorithm stores vectors in sparse representation; i.e., maintaining support of
the vector along with the values at corresponding positions.

As discussed in Section~\ref{sec:main:time}, each invocation of procedure $\proc{Search}$
takes $O(n)$ arithmetic operations, and procedure $\proc{Estimate}$ takes $O(r 2^r + k f(n))=O(2^r f(n))$ operations
(using naive sorting to find the largest coefficients and noting that $2^r \geq k$ and $f(n)=\Omega(n)=\Omega(r)$).  

We observe that for every $k$-sparse $w \in \R^N$, and $t \in [D]$, computing the multiplication $M^t \cdot k$ 
(which itself would be a $k$-sparse vector) takes $O(k f(n))$ operations ($k$
invocations of the condenser function, once for each nonzero entry of $w$, each time adding the
corresponding entry of $w$ to the correct position in the result vector).  Note that the indexing
time for updating an entry of the resulting vector is logarithmic in its length, which would be
$r \leq n$ and thus the required indexing time is absorbed into the above asymptotic since $f(n)=\Omega(n)$.
Moreover, we observe that without an effect in the above running time, we can in fact
compute $(M^t \otimes B) \cdot w$; since for each $i \in [N]$ on the support of $w$, the
corresponding $w(i)$ is added to a subset of the copies of $M^t$ depending on the bit representation
of $i$ and thus the additional computation per entry on the support of $w$ is $O(n)$, 
which is absorbed in the time $f(n)=\Omega(n)$ needed to compute the condenser function. Altogether we
see that computing $(M^t \otimes B) \cdot k$ can be done with $O(k f(n))$ arithmetic operations.

Since procedure $\proc{Recover}'$ loops $q$ times instead of $D$ times in each of the $s_0$ iterations,
each iteration taking time $O(2^r f(n))$, we see that the algorithm requires
$O(2^r q s_0 f(n))$ arithmetic operations in total. 
Now we can plug in the values of $q$ and $s_0$ by the analysis in the previous section and upper bound
the number of operations performed by the algorithm by 
\[ O(2^r \cdot \log(\log(1/\nu)/\eta) \cdot \log(1/\nu) f(n)). \]

This completes the running time analysis of the algorithm in Figure~\ref{fig:code:random}.

\oldLinespread
\bibliographystyle{abbrv}
\bibliography{\jobname} 
\newLinespread

\appendix

\section{Proof of Theorem~\ref{thm:GUVcondLinear} (construction of the lossless condenser)} \label{sec:proof:GUVcondLinear}

In this appendix, we include a proof of  Theorem~\ref{thm:GUVcondLinear} from \cite{ref:mahdiPhD}.
The first step is to recall the original framwork for construction of lossless condensers in \cite{ref:GUV09}
which is depicted in Construction~\ref{constr:GUV}. The construction is defined with respect to a prime power
alphabet size $q$ and integer parameter $u > 1$.

\begin{constr}
\caption{Guruswami-Umans-Vadhan's Condenser $C\colon \F_q^n \times
    \F_q \to \F_q^\ell$.}%
  \begin{itemize}
  \item \textit{Given:} A random sample $X \sim \cX$, where $\cX$ is a
    distribution on $\F_q^n$ with min-entropy at least $\kappa$, and a
    uniformly distributed random seed $Z \sim \U_{\F_q}$ over $\F_q$.

  \item \textit{Output:} A vector $C(X,Z)$ of length $\ell$ over $\F_q$.

  \item \textit{Construction:} Take any irreducible univariate
    polynomial $g$ of degree $n$ over $\F_q$, and interpret the input
    $X$ as the coefficient vector of a random univariate polynomial
    $F$ of degree $n-1$ over $\F_q$. Then, for an integer parameter
    $u$, the output is given by
    \[
    C(X,Z) := (F(Z), F_1(Z), \ldots, F_{\ell-1}(Z)),
    \]
    where we have used the shorthand $F_i := F^{u^i} \bmod g$.
  \end{itemize}
    \label{constr:GUV}
\end{constr}

The following key result about Construction~\ref{constr:GUV} is proved in \cite{ref:GUV09}:

\begin{thm} \cite{ref:GUV09} \label{thm:GUVcondGeneral} For any $\kappa > 0$, the mapping defined in
  Construction~\ref{constr:GUV} is a $(\kappa, \eps)$ lossless
  condenser with error $\eps := (n-1)(u-1)\ell/q$, provided that $\ell \geq
  \kappa/\log u$.
\end{thm}

By a careful choice of the parameters, the condenser can be made linear
as observed by Cheraghchi~\cite{ref:mahdiPhD}. We quote this result, which is
a restatement of Theorem~\ref{thm:GUVcondLinear}, below.

\begin{coro} \label{coro:GUVcondLinear} \cite{ref:mahdiPhD}
 Let $p$ be a fixed prime power and $\alpha > 0$ be an arbitrary constant. Then, for
 parameters $n \in \N$, $\kappa \leq n \log p$, and $\eps > 0$, there is an explicit $(\kappa, \eps)$-lossless condenser
 $h\colon \F_{p}^n \times \zo^d \to \F_{p}^r$ with 
  $d \leq (1+1/\alpha) (\log (n\kappa/\eps) + O(1))$ and output length satisfying
  $r \log p \leq d+(1+\alpha)\kappa$. Moreover, $h$ is a linear function (over $\F_{p}$) for every fixed choice of the second parameter.
\end{coro}

\begin{Proof}
  We set up the parameters of the condenser $C$ given by Construction~\ref{constr:GUV} and
  apply Theorem~\ref{thm:GUVcondGeneral}. The range of the parameters is mostly similar to what
  chosen in the original result of Guruswami et al.~\cite{ref:GUV09}.

  Letting $u_0 := (2p^2 n\kappa/\eps)^{1/ \alpha}$, we take $u$ to be an integer power of $p$ in range
  $[u_0, p u_0]$. Also, let $\ell := \lceil \kappa/\log u \rceil$ so that the condition $\ell \geq \kappa/\log u$ required by Theorem~\ref{thm:GUVcondGeneral}
  is satisfied. Finally, let $q_0 := nu \ell/\eps$ and choose the field size $q$ to be an integer power of $p$ in range
  $[q_0, p q_0]$.

  We choose the input length of the condenser $C$ to be equal to $n$. Note that
  $C$ is defined over $\F_q$, and we need a condenser over $\F_p$. Since $q$ is a power of $p$, 
  $\F_p$ is a subfield of $\F_q$.  For $x \in \F_p^n$ and $z \in \zo^d$,
  let $y := C(x,y) \in \F_q^\ell$, where $x$ is regarded as a vector over the extension $\F_q$ of $\F_p$.
  We define the output of the condenser $h(x,z)$ to be the vector $y$ regarded as a vector of length
  $\ell \log_p q$ over $\F_p$ (by expanding each element of $\F_q$ as a vector of length $\log_p q$ over $\F_p$).
  Clearly, $h$ is a $(\kappa, \eps)$-condenser if $C$ is.

  By Theorem~\ref{thm:GUVcondGeneral}, $C$ is a lossless condenser with error upper bounded by
  \[
   \frac{(n-1)(u-1)\ell }{q} \leq \frac{nu \ell}{q_0} = \eps.
  \]
  It remains to analyze the seed length $d$ and the output length $r$ of the condenser.
  For the output length of the condenser, we have
  \[
   r \log p = \ell \log q \leq (1+\kappa/\log u) \log q \leq d + \kappa (\log q)/(\log u),
  \]
  where the last inequality is due to the fact that we have $d = \lceil \log q \rceil$.
  Thus in order to show the desired upper bound on the output length, it suffices
  to show that $\log q \leq (1+\alpha) \log u_0$. We have
  \[
   \log q \leq \log (p q_0) = \log(pnu \ell/\eps) \leq \log u_0 + \log(p^2 n\ell/\eps)
  \]
  and our task is reduced to showing that $p^2 n \ell/\eps \leq u_0^{\alpha} = 2p^2n\kappa/\eps$. But
  this bound is obviously valid by the choice of $\ell \leq 1+ \kappa/\log u$.

  Now, $d = \lceil \log q \rceil$ for which we have
  \begin{eqnarray*}
   d &\leq& \log q + 1 \leq \log q_0 + O(1) \\
   &\leq& \log (nu_0 \ell / \eps) + O(1) \\
   &\leq& \log (nu_0 \kappa / \eps) + O(1) \\
   &\leq& \log(n\kappa/\eps) + \frac{1}{\alpha} \log (2p^2 n\kappa/\eps) \\
   &\leq& \big(1+ \frac{1}{\alpha}\big)(\log(n\kappa/\eps) + O(1))
  \end{eqnarray*}
as desired.

Since $\F_q$ has a fixed characteristic, an efficient deterministic algorithm for representation and
manipulation of the field elements is available \cite{ref:Shoup} which implies that the condenser
is polynomial-time computable and is thus explicit.

Moreover, since $u$ is taken as an integer power of $p$ and $\F_q$ is an extension of $\F_p$,
for any choice of polynomials $F, F', G \in \F_q[X]$, subfield elements $a, b \in \F_p$, and integer $i \geq 0$, we have
\[
 (a F + b F')^{u^i} \equiv a F^{u^i} + b F'^{u^i} \pmod G,
\]
meaning that raising a polynomial to power $u^i$ is an $\F_p$-linear operation.
Therefore, the mapping $C$ that defines the condenser (Construction~\ref{constr:GUV})
is $\F_p$-linear for every fixed seed. This in turn implies that the final condenser $h$
is linear, as claimed.
\end{Proof}

\section{The Leftover Hash Lemma} \label{sec:leftover}

Leftover Hash Lemma (first stated by Impagliazzo, Levin, and
Luby \cite{ref:ILL89}) is a basic and classical result in computational complexity
which is normally stated in terms of randomness extractors. However, it is
easy to observe that the same technique can be used to construct linear lossless condensers
with optimal output length (albeit large seed length). In other words, the lemma
shows that any universal family of hash functions can be turned into a linear
extractor or lossless condenser. For completeness, in this section we 
include a proof of this fact.
 

\newcommand{\cH}{\mathcal{H}}
\newcommand{\hLin}{\mathcal{H}_{\mathsf{lin}}}

\begin{defn} \label{defn:pwindep} \index{universal hash family}
  A family of functions $\cH = \{h_1, \ldots, h_D\}$ where $h_t\colon
  \zo^n \to \zo^r$ for $t=1, \ldots, D$ is called \emph{universal} if,
  for every fixed choice of $x, x' \in \zo^n$ such that $x \neq x'$
  and a uniformly random $t \in [D] := \{1, \ldots,
  D\}$\index{notation!$[n] := \{1,\ldots,n\}$} we have
  \[
  \Pr_t [h_t(x) = h_t(x')] \leq 2^{-r}.
  \]
\end{defn}

One of the basic examples of universal hash families is what we call
\emph{the linear family}, defined as follows. Consider an arbitrary
isomorphism $\varphi\colon \F_2^n \to \F_{2^n}$ between the vector
space $\F_2^n$ and the extension field $\F_{2^n}$, and let $0 < r \leq
n$ be an arbitrary integer.  The linear family $\hLin$ is the set $\{
h_\beta\colon \beta \in \F_{2^n} \}$ of size $2^n$ that contains a
function for each element of the extension field $\F_{2^n}$.  For each
$\beta$, the mapping $h_\beta$ is given by
\[
h_\beta(x) := (y_1, \ldots, y_r), \text{ where $(y_1, \ldots, y_n) :=
  \varphi^{-1}(\beta \cdot \varphi(x))$}.
\]
Observe that each function $h_\beta$ can be expressed as a linear
mapping from $\F_2^n$ to $\F_2^r$.  Below we show that this family is
pairwise independent.

\begin{prop}
  The linear family $\hLin$ defined above is universal.
\end{prop}

\begin{proof}
  Let $x, x'$ be different elements of $\F_{2^n}$. Consider the
  mapping $f\colon \F_{2^n} \to \F_2^r$ defined as
  \[
  f(x) := (y_1, \ldots, y_r), \text{ where $(y_1, \ldots, y_n) :=
    \varphi^{-1}(x)$},
  \]
  which truncates the binary representation of a field element from
  $\F_{2^n}$ to $r$ bits.  The probability we are trying to estimate
  in Definition~\ref{defn:pwindep} is, for a uniformly random $\beta
  \in \F_{2^n}$,
  \[
  \Pr_{\beta \in \F_{2^n}} [f(\beta \cdot x) = f(\beta \cdot x')] =
  \Pr_{\beta \in \F_{2^n}} [f(\beta \cdot (x-x')) = 0].
  \]
  But note that $x-x'$ is a nonzero element of $\F_{2^n}$, and thus,
  for a uniformly random $\beta$, the random variable $\beta x$ is
  uniformly distributed on $\F_{2^n}$. It follows that
  \[
  \Pr_{\beta \in \F_{2^n}} [f(\beta \cdot (x-x')) = 0] = 2^{-r},
  \]
  implying that $\hLin$ is a universal family.
\end{proof}

Now we are ready to state and prove the Leftover Hash Lemma (focusing on
the special case of lossless condensers).

\begin{thm} (Leftover Hash Lemma) \index{Leftover Hash
    Lemma} \label{lem:leftover} Let $\cH = \{ h_t\colon \F_2^n \to
  \F_2^r \mid t \in \F_2^d \}$ be a universal family of hash functions
  with $D$ elements, and  
  define the function $h\colon \F_2^n \times [D] \to \F_2^r $ as
  $h(x, t) := h_t(x)$. Then,
for every $\kappa, \eps$ such that $r \geq \kappa + 2 \log(1/\eps)$, the
    function $h$ is a $(\kappa,\eps)$-lossless condenser.
  In particular, by choosing $\cH = \hLin$, it is possible to get
  explicit extractors and lossless condensers with $D=2^n$.
\end{thm}

\begin{proof}
  Recall that by Definition~\ref{def:condenser} we need to show that for any distribution $\cX$
  over $\F_2^n$ and random variable $X$ drawn from $\cX$ and independent random variable $Z$ 
  uniformly drawn from $[D]$,
  respectively, the distribution of $h(X, Z)$ is $\eps$-close in statistical distance to a distribution
  with min-entropy at least $\kappa$.
  By a convexity argument, it suffices to show the
  claim when $\cX$ is the uniform distribution on a set $\supp(\cX)$ of size $K :=
  2^\kappa$ (on the other hand, we only use the lemma for such distributions in this paper). 
  
  Define $R := 2^r$, $D := 2^d$, and let $\mu$ be any 
  distribution uniformly supported on some set $\supp(\mu) \subseteq [D] \times \F_2^{r}$ such that $[D] \times \supp(\cX) \subseteq
  \supp(\mu)$, and denote by $\cY$ the distribution of $(Z, h(X, Z))$
  over $[D] \times \F_2^{r}$.  We will
  first upper bound the $\ell_2$ distance of the two distributions
  $\cY$ and $\mu$ (i.e., the $\ell_2$ difference of probability vectors defining the two distributions),
   that can be expressed as follows (we will use the notation $\cY(x)$ for the probability assigned to $x$ by $\cY$,
   and similarly $\mu(x)$):
  {\allowdisplaybreaks
    \begin{eqnarray}
      \| \cY - \mu \|_2^2 &=& \sum_{x \in [D] \times \F_2^{r}} (\cY(x) - \mu(x))^2 \nonumber \\
      &=& \sum_{x} \cY(x)^2 + \sum_{x} \mu(x)^2 -2 \sum_{x} \cY(x) \mu(x) \nonumber \\
      &\stackrel{\mathrm{(a)}}{=}& \sum_{x} \cY(x)^2 + \frac{1}{|\supp(\mu)|} -\frac{2}{|\supp(\mu)|} \sum_{x} \cY(x) \nonumber \\
      &=& \sum_{x} \cY(x)^2 - \frac{1}{|\supp(\mu)|} \label{eqn:LLLa},
    \end{eqnarray}}
  where $\mathrm{(a)}$ uses the fact that $\mu$ assigns probability $1/|\supp(\mu)|$
  to exactly $|\supp(\mu)|$ elements of $[D] \times \F_2^{r}$ and zeros elsewhere.

  Now observe that $\cY(x)^2$ is the probability that two independent
  samples drawn from $\cY$ turn out to be equal to $x$, and thus,
  $\sum_{x} \cY(x)^2$ is the \emph{collision probability} of two
  independent samples from $\cY$, which can be written as
  \begin{equation*}
    \sum_{x} \cY(x)^2 = \Pr_{Z,Z',X,X'}[(Z, h(X, Z)) = (Z', h(X', Z'))],
  \end{equation*}
  where the random variables $Z, Z'$ are uniformly and independently sampled from $[D]$ and $X, X' $ are independently
  sampled from $\cX$. We can rewrite the collision probability as
  {\allowdisplaybreaks
    \begin{eqnarray*}
      \sum_{x} \cY(x)^2 &=& \Pr[Z = Z'] \cdot \Pr[ h(X, Z) = h(X', Z') \mid Z = Z'] \\
      &=& \frac{1}{D} \cdot \Pr_{Z,X,X'}[ h_{Z}(X) = h_{Z}(X') ] \\
      &=& \frac{1}{D} \cdot (\Pr[X=X'] + \frac{1}{K^2} \sum_{\substack{x, x' \in \supp(\cX) \\ x \neq x'}} \Pr_{Z}[ h_{Z}(x) = h_{Z}(x') ]) \\
      &\stackrel{\mathrm{(b)}}{\leq}& \frac{1}{D} \cdot \big(\frac{1}{K} + \frac{1}{K^2} \sum_{\substack{x, x' \in \supp(\cX) \\ x \neq x'}} \frac{1}{R}\big)
      \leq \frac{1}{DR} \cdot \big(1 + \frac{R}{K}\big),
    \end{eqnarray*}}
  \noindent where $\mathrm{(b)}$ uses the assumption that $\cH$ is a
  universal hash family.  Plugging the bound in \eqref{eqn:LLLa}
  implies that
  \[
  \| \cY - \mu \|_2 \leq \frac{1}{\sqrt{DR}} \cdot \sqrt{1 -
    \frac{DR}{|\supp(\mu)|} + \frac{R}{K}}.
  \]
  Observe that both $\cY$ and $\mu$ assign zero probabilities to
  elements of $[D] \times \F_2^{r}$ outside the support of $\mu$. Thus using
  Cauchy-Schwarz on a domain of size $|\supp(\mu)|$, the above bound
  implies that the statistical distance between $\cY$ and $\mu$ is at
  most
  \begin{equation} \label{eqn:LLLb} \frac{1}{2} \| \cY - \mu \|_1 \leq \frac{1}{2} \cdot
    \sqrt{\frac{|\supp(\mu)|}{DR}} \cdot \sqrt{1 -
      \frac{DR}{|\supp(\mu)|} + \frac{R}{K}}.
  \end{equation}
%
  Now, we specialize
  $\mu$ to any distribution that is uniformly supported on a set of size $DK$ containing
  $\supp(\cY)$ (note that, since $\cX$ is assumed to be uniformly distributed on its support,
  $\cY$ must have a support of size at most $DK$). Since
  $r \geq \kappa + 2\log(1/\eps)$, we have $K = \eps^2 R$, and
  \eqref{eqn:LLLb} implies that $\cY$ and $\mu$ are $\eps$-close (in fact, $(\eps/2)$-close) in statistical distance.
\end{proof}

\end{document}